\documentclass[conference,10pt]{IEEEtran}
\IEEEoverridecommandlockouts
\usepackage{ifthen}
\newboolean{conference}
\newboolean{only_done}

\usepackage[utf8]{inputenc}
\usepackage{latexsym, amssymb, dsfont}
\usepackage{amsmath,amsthm}
\usepackage{algorithm}
\usepackage{algpseudocode}
\usepackage{siunitx}

\usepackage[acronym]{glossaries}
\glsdisablehyper
\setacronymstyle{long-short}
\newacronym{idm}{IDM}{Involution Delay Model}
\newacronym{ddm}{DDM}{Degradation Delay Model}
\newacronym{cidm}{CIDM}{Composable Involution Delay Model}
\newacronym{tct}{TCT}{Threshold Crossing Times}
\newacronym{wst}{WST}{Waveform Switching Times}

\usepackage{url}

\usepackage{tikz}
\usetikzlibrary{circuits.logic.US}
\pgfdeclarelayer{background} \pgfsetlayers{background,main}
\usetikzlibrary{arrows,shapes,decorations,decorations.pathreplacing,decorations.pathmorphing,intersections,calc,fit}
\usepackage{hyperref}
\usepackage[capitalize]{cleveref}
\crefname{figure}{Fig.}{Figs.}
\crefname{equation}{}{}

\newcounter{lemcounter}
\newtheorem{theorem}{Theorem}
\newtheorem{lemma}[theorem]{Lemma}

\newtheorem{definition}[theorem]{Definition}

\newtheorem{observation}[theorem]{Observation}

\newcommand{\dup}{\delta_\uparrow}
\newcommand{\ddo}{\delta_\downarrow}
\newcommand{\bdup}{\overline{\delta}_\uparrow}
\newcommand{\bddo}{\overline{\delta}_\downarrow}
\newcommand{\bT}{\overline{T}}
\newcommand{\bdmin}{\overline{\delta}_{min}}
\newcommand{\fup}{f_\uparrow}
\newcommand{\fdo}{f_\downarrow}

\newcommand{\dmin}{\delta_{min}}
\newcommand{\rdmin}{\delta_{min}}

\newcommand{\dinfty}{\delta_{\infty}}

\newcommand{\dudmin}{\delta^{\uparrow/\downarrow}_{min}}
\newcommand{\Deltapm}{\Delta^{+/-}}
\newcommand{\dupmin}{\delta^\uparrow_{min}}
\newcommand{\ddomin}{\delta^\downarrow_{min}}
\newcommand{\dupzero}{\delta_\uparrow^{0}}
\newcommand{\ddozero}{\delta_\downarrow^{0}}

\newcommand{\spice}{SPICE}

\newcommand{\idmm}{IDM*}
\newcommand{\idmf}{IDM+}

\newcommand{\vth}{V_{th}}

\newcommand{\vthin}{V_{th}^{in}}

\newcommand{\vthout}{V_{th}^{out}}
\newcommand{\vthinm}{V_{th}^{in*}}
\newcommand{\vthoutm}{V_{th}^{out*}}

\newcommand{\vdd}{V_{DD}}
\newcommand{\gnd}{\texttt{GND}}
\newcommand{\vout}{V_{out}}
\newcommand{\vin}{V_{in}}

\newcommand{\C}{{\cal C}}

\newcommand{\bfno}[1]{\noindent{\bf #1}}

\newcommand{\thresholder}{T\!h}

 \usepackage{tikz}
 \usepackage{color}
 \usepackage{tikz-timing}
 \usepackage{pgfplots}
 \pgfplotsset{compat=1.13}
 \usepgfplotslibrary{units}
 \usetikzlibrary{intersections}

\usepackage{pgfplotstable}
\usepackage{fp}
\usepackage{circuitikz}
\usetikzlibrary{automata}
\usetikzlibrary{arrows,shapes.gates.logic.US,shapes.gates.logic.IEC,calc}
\usetikzlibrary{shapes,snakes}
\usetikzlibrary{calc}
\usetikzlibrary{matrix}
\usetikzlibrary{positioning,calc}
\usetikzlibrary{through}
\usetikzlibrary{spy, fit}
\usetikzlibrary{decorations.markings}
\usetikzlibrary{pgfplots.groupplots}

\definecolor{myBlue}{HTML}{4DB8FF}
\definecolor{myBlueFill}{HTML}{CDE8FF}
\definecolor{myBlueWrite}{HTML}{1D78FF}
\definecolor{myGreen}{HTML}{5CD65C}
\definecolor{myGreenFill}{HTML}{CCFFCC}
\definecolor{myGreenWrite}{HTML}{2CA62C}
\definecolor{myYellow}{HTML}{F5C134}
\definecolor{myYellowFill}{HTML}{FFDD99}
\definecolor{myYellowWrite}{HTML}{B37700}

\definecolor{color1}{HTML}{E41A1C}
\definecolor{color2}{HTML}{377EB8}
\definecolor{color3}{HTML}{4DAF4A}
\definecolor{color4}{HTML}{984EA3}
\definecolor{color5}{HTML}{FF7F00}

\tikzstyle{signal}=[line width=1.6pt, color=myBlueWrite]

\tikzstyle{stable}=[circle, fill=black,inner sep=2.5pt]
\tikzstyle{metastable}=[star, fill=black,inner sep=2.5pt]

\author{
    \IEEEauthorblockN{
      J\"urgen Maier\IEEEauthorrefmark{1}
      \begin{minipage}[c]{1em}
        \href{https://orcid.org/0000-0002-0965-5746}{\includegraphics[width=1em]{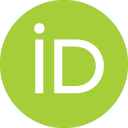}}
      \end{minipage}\,,
      Daniel \"Ohlinger\IEEEauthorrefmark{1}
      \begin{minipage}[c]{1em}
        \href{https://orcid.org/0000-0001-8097-3619}{\includegraphics[width=1em]{orcID.png}}
      \end{minipage}\,,
      Ulrich Schmid\IEEEauthorrefmark{1}
      \begin{minipage}[c]{1em}
        \href{https://orcid.org/0000-0001-9831-8583}{\includegraphics[width=1em]{orcID.png}}
      \end{minipage}\,,
      Matthias F\"ugger\IEEEauthorrefmark{2}
      \begin{minipage}[c]{1em}
        \href{https://orcid.org/0000-0001-5765-0301}{\includegraphics[width=1em]{orcID.png}}
      \end{minipage}\,,
      Thomas Nowak\IEEEauthorrefmark{3}
      \begin{minipage}[c]{1em}
        \href{https://orcid.org/0000-0003-1690-9342}{\includegraphics[width=1em]{orcID.png}}
      \end{minipage}\,
    }
    \IEEEauthorblockA{\IEEEauthorrefmark{1}ECS Group, TU Wien\\
    daniel.oehlinger@tuwien.ac.at, \{jmaier, s\}{@}ecs.tuwien.ac.at}
    \IEEEauthorblockA{\IEEEauthorrefmark{2}CNRS \& LSV, ENS Paris-Saclay, 
    Universit\'e Paris-Saclay \& Inria\\
    mfuegger@lsv.fr}
    \IEEEauthorblockA{\IEEEauthorrefmark{3}Universit\'e Paris-Saclay, CNRS\\
    thomas.nowak{@}lri.fr}
}

\title{A Composable Glitch-Aware Delay Model \thanks{This research was partially
    funded by the Austrian Science Fund (FWF) projects DMAC (P32431) and ADynNet
    (P28182), the Centre National de la Recheche Scientifique projects ABIDE and
    BACON, the Agence Nationale de la Recherche project FREDDA
    (ANR-17-CE40-0013), and by the DigiCosme working group HicDiesMeus under ANR
    grant agreements ANR-11-LABEX-0045-DIGICOSME and ANR-11-IDEX-0003-02} }

\begin{document}

\maketitle%

\setlength\leftmargini{10pt}


\setboolean{conference}{false}  
\setboolean{only_done}{true}  


\begin{abstract}
  We introduce the Composable Involution Delay Model (CIDM) for fast and
  accurate digital simulation. It is based on the Involution Delay Model (IDM)
  [F\"ugger et~al., IEEE TCAD 2020], which has been shown to be the only
  existing candidate for faithful glitch propagation known so far.  In its
  present form, however, it has shortcomings that limit its practical
  applicability and utility. First, IDM delay predictions are conceptually based
  on discretizing the analog signal waveforms using specific matching input and
  output discretization threshold voltages. Unfortunately, they are difficult to
  determine and typically different for interconnected gates. Second,
  metastability and high-frequency oscillations in a real circuit could be
  invisible in the IDM signal predictions. Our CIDM reduces the characterization
  effort by allowing independent discretization thresholds, improves
  composability and increases the modeling power by exposing canceled pulse
  trains at the gate interconnect. We formally show that, despite these
  improvements, the CIDM still retains the IDM's faithfulness, which is a
  consequence of the mathematical properties of involution delay functions.

\end{abstract}



\section{Introduction and Context}
\label{sec:intro}

Accurate prediction of signal propagation is a crucial task in modern digital
circuit design. Although the highest precision is obtained by analog
simulations, e.g., using \spice, they suffer from excessive simulation
times. Digital timing analysis techniques, which rely on (i) discretizing the
analog waveform at certain thresholds and (ii) simplified interconnect
resp. gate delay models, are hence utilized to verify most parts of a circuit.
Prominent examples of the latter are pure (constant input-to-output delay
$\Delta$) and inertial delays (constant delay $\Delta$, pulses shorter than an
upper bound are removed) \cite{Ung71}. To accurately determine $\Delta$, which
stays constant for all simulation runs, highly elaborate estimation methods like
CCSM~\cite{Syn:CCSM} and ECSM~\cite{Cad:ECSM} are required.

Single-history delay models, like the \emph{\gls{ddm}} \cite{BJV06}, have been
proposed as a more accurate alternative.  Here, the input-to-output delay
$\delta(T)$ depends on a single parameter, the previous-output-to-input delay
$T$ (see \cref{fig:single_history}).
  Still, F\"ugger \textit{et~al.}\ \cite{FNS16:ToC} showed that none of the
  existing delay models, including DDM, is faithful, as they fail to correctly
  predict the behavior of circuits solving the canonical \emph{short-pulse
    filtration} (SPF) problem.

  F\"ugger \textit{et~al.}~\cite{FNNS19:TCAD} introduced the
  \textit{\gls{idm}} and showed that, unlike all other digital delay models
  known so far, it faithfully models glitch propagation.  The distinguishing
  property of the \gls{idm} is that the delay functions for rising ($\dup$) and
  falling ($\ddo$) transitions form an involution, i.e., $-\dup(-\ddo(T))=T$.  A
  simulation framework based on ModelSim, the Involution Tool, confirmed the
  accuracy of the model predictions for several simple
  circuits~\cite{OMFS20:INTEGRATION}.

Nevertheless, the \gls{idm} shows, at the moment, several shortcomings which
impair its composability:

\noindent
(I) Ensuring the involution property requires specific (``matching'')
\emph{discretization threshold voltages} $\vthinm$ and $\vthoutm$ to discretize
the analog waveforms at in- and output. These are not only unique for every
single gate in the circuit but also difficult to determine.
(II) The discretization threshold voltages may vary from gate to gate,
  i.e., the matching output discretization threshold voltage $\vthoutm$ of a
  given gate $G_1$ and the matching input discretization threshold voltage
  $\vthinm$ of a successor gate $G_2$ are not necessarily the same.
  Consequently, just adding the delay predictions of the IDM for $G_1$ and $G_2$
  cannot be expected to accurately model the overall delay of their
  composition. This is particularly true for circuits designs where different
  transistor threshold voltages \cite{Ase98} are used for tuning the delay-power
  trade-off \cite{NBPP10} or reliability \cite{IB09}.

\noindent
(III) Intermediate voltages, caused by creeping or oscillatory metastability, are
expressed differently for various values of $\vthoutm$:
Ultimately, a single analog trajectory may result in either zero, one or a whole
train of digital transitions.

\begin{figure}[t]
  \centering
  \scalebox{0.8}{\begin{tikzpicture}[>=stealth, scale=0.8]

	\def\lw{1.9pt}
  \draw[->] (0,0) -- (0,1.2) node [anchor=east] {in($t$)};
  \draw[->] (0,0) -- (10,0) node [anchor=south] {$t$};
  \draw[line width=\lw, color=color2] (0,0.9) -- (0.3,0.9) -- (0.3,0.1) -- (5,0.1) -- (5,0.9) -- (10,0.9);

  \draw[->] (0,-1.5) -- (0,-0.3) node [anchor=east] {out($t$)};
  \draw[->] (0,-1.5) -- (10,-1.5) node [anchor=south] {$t$};
  \draw[line width=\lw, color=color3] (0,-0.6) -- (3.5,-0.6) -- (3.5,-1.4) -- (7.5,-1.4) --
  (7.5,-0.6) -- (10,-0.6);

  \draw[dashed] (3.5,-1.4) -- (3.5,-2.4);
  \draw[dashed] (5,0) -- (5,-2.4);
  \draw[dashed] (7.5,-1.4) -- (7.5,-2.4);
  \draw[densely dashed,->,shorten >=1pt,shorten <=2pt,thick] (0.3,0.9) -- (3.45,-0.55);
  \draw[densely dashed,->,shorten >=1pt,shorten <=2pt,thick] (5,0.9) -- (7.5,-0.6);

  \node (T) at (4.25,-2) {$T$};
  \node (DT) at (6.25,-2) {$\dup (T)$};
  \draw[-] (T) -- (3.5,-2);
  \draw[->] (T) -- (5,-2);

  \draw[->] (DT) -- (7.5,-2);
  \draw[-] (DT) -- (5,-2);
\end{tikzpicture}}
  \caption{The delay $\dup$ as function of $T$. Taken
    from~\cite{OMFS20:INTEGRATION}.}
  \label{fig:single_history}
\end{figure}
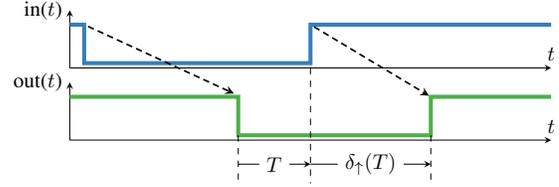

\noindent
\medskip
\textbf{Main contributions:}

In the present paper, we address all these
shortcomings. 
\begin{enumerate}
\item[1)] Based on the empirical analysis of the impact of $\vthin$ 
and $\vthout$ on the delay functions of a real circuit, we introduce
CIDM channels, which are essentially made up of a typically asymmetric 
(for rising/falling transitions)
pure delay shifter followed by an IDM channel (a PI channel).
CIDM channels can model gates with arbitrarily discretization
threshold $\vthinm$ and $\vthoutm$ and simplify the characterization
of their delay functions. Moreover, the CIDM exposes cancellations at 
the interconnect between gates, rather than removing them silently as in the IDM. 
Consequently, digital timing simulators
can, for example, record trains of canceled pulses and report it to the user
accordingly.

\item[2)] We prove that CIDM channels are \emph{not} equivalent to IDM channels per se, in
the sense that their delay functions are not involutions. However, the fact that
both essentially contain the same components allows us to map a CIDM circuit
description to an IDM one: For every circuit modeled with compatible CIDM 
channels, modeling matching discretization threshold voltages, there
is an equivalent decomposition into IP channels, consisting of an IDM channel
followed by an arbitrary pure delay shifter. Somewhat surprisingly, the
mathematical properties of involution delay functions reveal that IP channels
are instances of IDM channels. Consequently, the faithfulness of the IDM carries
over literally to the CIDM.

\item[3)] We present the theoretical foundations and an implementation of a simulation
  framework
for CIDM, which was incorporated into the Involution Tool~\cite{OMFS20:INTEGRATION}.
\footnote{The original
Involution Tool is accessible via \url{https://github.com/oehlinscher/InvolutionTool},
our extended version is provided at \url{https://github.com/oehlinscher/CDMTool}.} 
We also provide a proof of correctness of our simulation algorithm, which shows
that it always terminates after having produced the unique execution of an
arbitrary circuit composed of gates with compatible CIDM channels.

\item[(4)] We conducted a suite of experiments, both on a custom inverter chain with varying matching threshold voltages and a standard inverter chain. In the former case, we observed an impressive improvement of modeling accuracy for CIDM over IDM in many cases.
\end{enumerate}

\noindent
\textbf{Paper organization:} We start with some basic properties of the existing
\gls{idm} in \cref{sec:idm}. In \cref{sec:vth}, we empirically analyze the
impact of changing $\vthoutm$ and $\vthinm$ on the delay
functions. \cref{sec:gidm} introduces and justifies the CIDM, while
\cref{sec:gidmisidm} proves its faithfulness.  In \cref{sec:algorithm}, we
describe the CIDM simulation algorithm and its implementation,
and prove that its correctness. \cref{sec:experiments} provides the
results of our experiments. Some conclusion in \cref{sec:conclusion} round-off
our paper.

\section{Involution Delay Model Basics}
\label{sec:idm}

In this section, we briefly discuss the IDM.  For further details, the interested
reader is referred to the original publication~\cite{FNNS19:TCAD}.

The essential benefit of using delay functions which are involutions is their
ability to perfectly cancel zero-time input glitches: In
\cref{fig:single_history}, it is apparent that the rising input transition
causes a rising transition at the output, after delay $\dup(T)$. Now suppose
that there is an additional falling input transition immediately after the
rising one, actually occurring at the same time. Since this constitutes a
zero-time input glitch, which is almost equivalent to no pulse at all, the
output should behave as if it was not there at all.

For this purpose it is required that the delay of the additional falling input
transition to go \emph{back in time}, i.e., to exactly hit the predicted time of
the previous falling output transition: Note carefully that just canceling the
rising output transition, by generating the falling output transition at or
before the rising one, would not suffice, as the calculation of the parameter
$T$ for the next transition depends on the time of the previous output
transition. It is not difficult to check that this going back is indeed achieved
when $-\dup(-\ddo(T))=T$ and $-\ddo(-\dup(T))=T$ holds. As the IDM also requires
the delay functions to be strictly increasing and concave, the involution
property enforces them to be symmetric w.r.t.\ the 2\textsuperscript{nd} median
$y=-x$.

Lemma~3 in \cite{FNNS19:TCAD}, restated as \cref{lem:delta:min:idm} below, shows
that \emph{strictly causal} involution channels, characterized by strictly
increasing, concave delay functions with $\dup(0)>0$ and $\ddo(0)>0$, give raise
to a unique $\dmin>0$ that (i) resides on the 2\textsuperscript{nd} median
$y=-x$ and (ii) is shared by $\dup$ and $\ddo$ due to the involution property.

\begin{lemma}[{\cite[Lem.~3]{FNNS19:TCAD}}]
\label{lem:delta:min:idm}
A strictly causal involution channel has a unique~$\dmin>0$ defined by
$\dup(-\dmin) = \dmin = \ddo(-\dmin)$.
\end{lemma}

In \cite{FNNS19:TCAD}, F\"ugger \textit{et~al.}\ have shown that
self-inverse delay functions arise naturally in a (generalized) standard analog
model that consists of a pure delay, a slew-rate limiter with generalized
switching waveforms and an ideal comparator (see
Fig.~\ref{fig:analogy}). First, the binary-valued input~$u_i$ is
delayed by $\dmin>0$, which assures causal
channels, i.e., $\delta_{\uparrow/\downarrow}(0)>0$. For every transition on
$u_d$, the generalized slew-rate limiter switches to the
corresponding waveform ($\fdo/\fup$ for a falling/rising transition).
Note that the value at $u_r$ (the analog output voltage) does
not jump, i.e., is a continuous function. Finally, the comparator generates the
output $u_o$ by discretizing the value of this waveform w.r.t.\ the
discretization threshold $\vthout$.

Using this representation, the need for $\dup(T), \ddo(T) < 0$ can be
explained by the necessity to cover sub-threshold pulses, i.e., ones that do not
reach the output discretization threshold. In this case, the switching waveform
has to be followed into the past to cross $\vthout$, resulting in the seemingly
acausal behavior.

\begin{figure}[tb]
  \centerline{
    \begin{tikzpicture}[
cmp/.style={draw, isosceles triangle, isosceles triangle stretches, minimum width=1.1cm, minimum height=1.0cm},
xscale=0.8,
yscale=0.5,
>=latex'
]

\def\figcirampin{38}
\def\figcirampinlen{0.25}
\def\figcirampvoltlen{0.27}
\def\figcirampcap{0.3}
\def\figcirdista{0.45}
\def\figcirdistb{0.55}
\def\figcirdistc{1}
\def\figciroutlen{0.25}

\begin{scope}[color=black!40,text=black!40]
\path
(0,0) node (cmp1) [cmp] {};
\path [draw]
(cmp1.{180-\figcirampin}) node [anchor=west] {\footnotesize $+$} --
++(-\figcirampinlen,0) node (cmp1-west) [coordinate] {};
\path [draw]
(cmp1.{180+\figcirampin}) node [anchor=west] {\footnotesize $-$} -- ++(-\figcirampinlen,0);
\end{scope}

\path [draw, shorten <=-0.5pt]
(cmp1.east) -- node [anchor=south] {\color{color5} \large ${u_i}$}
++(\figcirdista,0) node (del) [anchor=west, draw, minimum height=0.4cm, minimum width=0.8cm,
                      inner sep=1pt, rounded corners=0.2cm] {\small $\dmin$};

\path [draw]
(del.east) -- node [anchor=south] {\color{color2} \large  $u_d$}
++(\figcirdistb,0) node (srl) [draw, minimum size=0.8cm, anchor=west] {};
\path [draw, scale=0.45]
(srl) ++(-0.5,-0.5) .. controls +(0.3,0) and +(-0.8,0) .. +(1,1);

\path [draw]
(srl.east) -- node [anchor=south] {\color{color3} \large  $u_r$}
++(\figcirdistc*1.3,0) node (cmp2) [cmp, anchor={180-\figcirampin}] {}
                     node [anchor=west] {\footnotesize $+$};
\path [draw]
(cmp2.{180+\figcirampin}) node [anchor=west] {\footnotesize $-$} --
++(-\figcirampinlen,0) node [anchor=east, inner sep=1pt] {\small $\vthout$};
\path [draw, shorten <=-0.5pt]
(cmp2.east) --
++(\figciroutlen,0) node [anchor=290] {\color{color4} \large $u_o$};

\foreach \cmpx in {cmp2} {
\path [draw]
(\cmpx.50) -- ++(0,\figcirampvoltlen) node (\cmpx-1v) [anchor=south] {\small $\vdd$};
\path [draw]
(\cmpx.310) -- ++(0,-\figcirampvoltlen) node (\cmpx-0v) [anchor=north] {\small $\gnd$};
\path
($(\cmpx)+(\figcirampcap,0)$) node {\small $\infty$};
}

\begin{scope}[black!40]
\foreach \cmpx in {cmp1} {
\path [draw]
(\cmpx.50) -- ++(0,\figcirampvoltlen) node (\cmpx-1v) [anchor=south] {\small $\vdd$};
\path [draw]
(\cmpx.310) -- ++(0,-\figcirampvoltlen) node (\cmpx-0v) [anchor=north] {\small $\gnd$};
\path
($(\cmpx)+(\figcirampcap,0)$) node {\small $\infty$};
}
\end{scope}

%
\end{tikzpicture}}
  \centerline{
    \begin{tikzpicture}[
xscale=5.4,
yscale=1.5,
>=latex'
]

\def\figanvth{0.4}
\def\figaninstart{0.1}
\def\figaninlen{0.28}
\def\figandelay{0.05}
\def\figanfup[#1]{(1-exp(-(#1)/0.15))*(1-exp(-(#1)/0.16))}
\def\figanfdown[#1]{(1-(1-exp(-((#1)^2)/0.01))*(1-exp(-((#1)^2)/0.1)))}
\def\figancaptpos{0.9}

\coordinate (origin) at (0,0);

\path [draw,->]
(0,0) -- (0,1.1) node (xaxis) [left] {\small $u(t)$};
\path [draw,->]
(0,0) -- (1.05,0)   node (yaxis) [above] {\small $t$};

\path [draw, dashdotted, very thin, name path=vth]
(0,\figanvth) node [left] {\small $\vth$} -- (1,\figanvth);

\path [draw, line width=1pt, color=color5]
(\figaninstart,0) -- ++(0,1) node [anchor=south east, xshift=3pt] {\large \color{color5} $u_i$} --
++(\figaninlen,0) -- ++(0,-1) node (dt3) [coordinate] {};

\path [draw, line width=1pt, color=color2]
({\figaninstart+\figandelay},0) -- ++(0,1) node [above, xshift=3pt, yshift=-1pt] {\large \color{color2} $u_d$} --
++(\figaninlen,0) -- ++(0,-1);

\path [draw, thick, xshift={(\figaninstart+\figandelay)*1cm}, name path=fup, color=color3, line width=1.5pt]
plot[smooth, domain=0:\figaninlen] (\x,{\figanfup[\x]});
\path [draw, thick, dotted, xshift={(\figaninstart+\figandelay)*1cm}, name path=fupdot]
plot[smooth, domain=\figaninlen:{1-\figaninstart-\figandelay}] (\x,{\figanfup[\x]});
\path [name path=func]
plot[smooth, domain=0:1] (\x,{\figanfdown[\x]});
\path [name path=const]
(0,{\figanfup[\figaninlen]}) -- +(1,0);
\path [name intersections={of=func and const}]
(intersection-1) node (finversept) [coordinate] {};
\newdimen\finverse
\pgfextractx{\finverse}{\pgfpointanchor{finversept}{center}}
\pgfmathsetlengthmacro{\figanfdownshift}{(\figaninstart+\figandelay+\figaninlen)*1cm-\finverse}
\path [draw, thick, dotted, xshift=\figanfdownshift]
plot[smooth, domain=0:{\finverse/1cm}] (\x,{\figanfdown[\x]}) ;
\path [draw, thick, xshift=\figanfdownshift, name path=fdown, color=color3, line width=1.5pt]
plot[smooth, domain={\finverse/1cm}:{(1cm-\figanfdownshift)/1cm}] (\x,{\figanfdown[\x]});

\path [draw, color=color4, line width=1pt, name intersections={of=fup and vth, by={up}}, name intersections={of=fdown and vth, by={down}}]
(origin -| up) node (dt2) [coordinate] {} -- ++(0,1) -|  node [above, xshift=3pt, yshift=-1pt] {\color{color4} \large $u_o$}
(origin -| down) node (dt4) [coordinate] {};

\path [name path=fcapts]
(\figancaptpos,0) -- ++(0,1);
\path [name intersections={of=fcapts and fdown, by={fd}}, name intersections={of=fcapts and fupdot, by={fu}}]
(fd) node [above] {\color{color3} \large $u_r$}
(fu) node [below] (tfu) {\small $f_\uparrow$};

\draw[help lines] (tfu.west) -- ++(-0.06,0.15);
\node (tfd) at (0.22,0.7) {\small $f_\downarrow$};
\draw[help lines] (tfd.east) -- ++(0.1,0.15);

\foreach \dt/\name in {{(\figaninstart,0)/1}, {(dt2)/2}, {(dt3)/3}, {(dt4)/4}} {
\path [draw, very thin, shorten <= 1pt, shorten >=-2pt]
\dt -- ++(0,-0.1) node (smth-\name) [coordinate] {};
}

\path [draw, <-, very thin]
(smth-1) -- node [anchor=north] {\tiny $T_1$} (smth-1 -| origin);
\path [draw, ->, very thin, shorten <=0.6pt]
(smth-1) -- node [anchor=north] {\tiny $\delta_\uparrow(T_1)$} (smth-2);
\path [draw, <-, very thin, shorten <=0.6pt]
(smth-3) -- node [anchor=north] {\tiny $T_2$} (smth-2);
\path [draw, ->, very thin, shorten <=0.6pt]
(smth-3) -- node [anchor=north] {\tiny $\delta_\downarrow(T_2)$} (smth-4);
\end{tikzpicture}}
  \caption{Analog channel model representation (upper part) with a sample execution
    (bottom part). Adapted from~\cite{FNNS19:TCAD}.}
  \label{fig:analogy}
\end{figure}
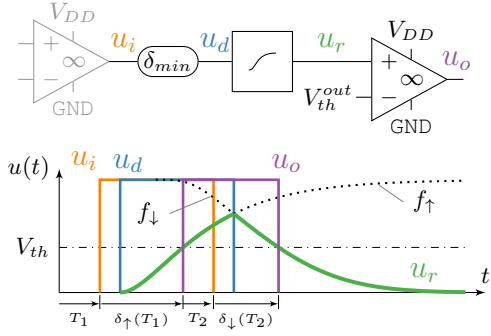

\section{Discretization Threshold Voltages}
\label{sec:vth}

In this section, we will empirically explore the relation of gate delays and
discretization threshold voltages by means of simulation results.\footnote{Our
  simulations have been performed for a buffer in the \SI{15}{\nm} NanGate
  library.  However, since we are only reasoning about qualitative aspects
  common to all technologies, the actual choice has no significance.  }  In most
of the following observations, we assume that a given physical (analog) gate is
to be characterized as a zero-time Boolean gate with a succeeding IDM channel
that models the delay. In order to accomplish concrete values, discretization
threshold voltages $\vthin$ at the input and $\vthout$ at the output of a gate
have to be fixed, and the pure delay component $\dmin$ of the IDM channel as well as
the delay functions $\dup$ and $\ddo$ are determined accordingly.

\begin{definition}
  The input and output discretization voltages $\vthin$ and $\vthout$ are called
    \emph{matching} for a gate, if
    the induced delay functions $\dup(T)$, $\ddo(T)$ fulfill the condition
    $\dup(-\dmin)= \dmin=\ddo(-\dmin)$.
  To stress that a pair of input and output discretization
    threshold voltages is matching, they will be denoted
    as $\vthinm$ and $\vthoutm$.
\end{definition}

We will now characterize properties of matching discretization threshold
voltages. They depend on many factors, including
transistor threshold voltages~\cite{Ase98} and the symmetry of the pMOS vs.\
nMOS stack.  Since varying these parameters is commonly used in
advanced circuit design to trade delay for
power consumption~\cite{NBPP10,GG14} and reliability~\cite{IB09}, as well as for
implementing special gates (e.g.\ logic-level conversion \cite{SRMS98}), the range of suitable discretization threshold voltages
could differ significantly among gates.

\begin{figure}
  \centering
  \includegraphics[width=0.9\columnwidth]{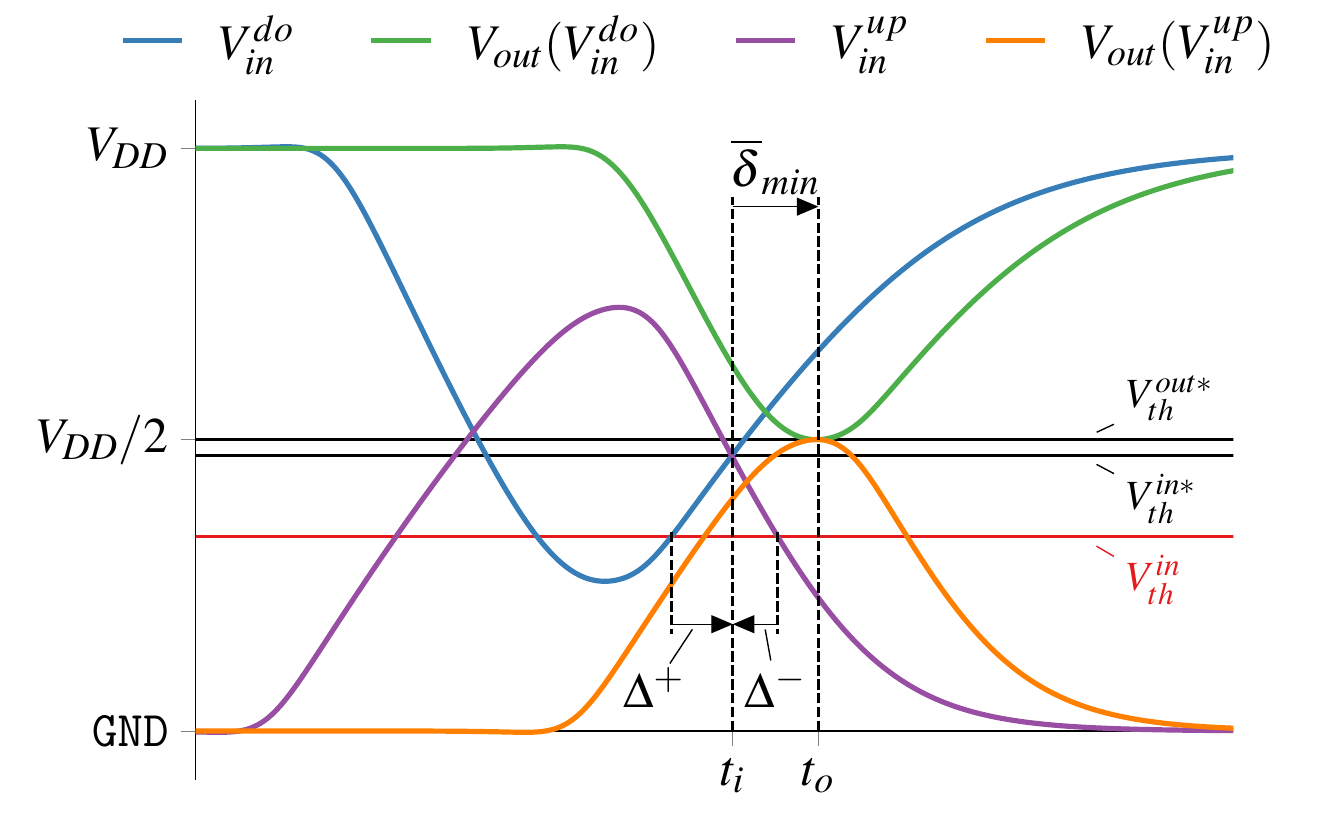}
  \caption{The relationship among $\vthinm$, $\dmin$ and $\vthoutm$.}
  \label{fig:im_char}
\end{figure}

Considering these circumstances it seems impossible that output and input
descretization values coincide among connected gates. However, the following
observation shows that there is an unlimited number of matching discretization
threshold pairs for \gls{idm}:

\begin{observation}
  \label{lem:vth:inout}
  For every choice of $\vthin$, there is exactly one matching $\vthout$.
  Fixing either of them uniquely determines the other and, in addition, also the
    pure delay $\dmin$.
\end{observation}

\begin{proof}[Justification]
  Let us fix $\vthout$ and investigate how $\vthin$ and $\dmin$ can be
  determined.  For this purpose, we consider an analog pulse at $\vout$
  that barely touches $\vthout$, i.e., results in a zero-time glitch in the
  digital domain.  There is a \emph{unique} positive and a \emph{unique}
  negative analog output pulse with this shape, which is both confirmed by
  simulation results and analytic considerations on the underlying system of
  differential equations \cite{M17:TR}.  Now shift the positive and negative
  pulses in time such that their output voltages touch $\vthout$, one from below
  and the other from above, at time $t_o$ (see \cref{fig:im_char}). Due to the
  condition $\ddo(-\dmin)=\dmin=\dup(-\dmin)$, this implies that the falling
  transition of the positive pulse and the rising transition of the negative
  pulse at the input must both cross $\vthin$ at time $t_i = t_o - \dmin$.
  Thus, fixing $\vthoutm$ uniquely determines the matching $\vthinm$ and
  $\dmin = t_o - t_i$.
\end{proof}

Actually determining the matching $\vthoutm$ for a given $\vthinm$ and vice versa is a
challenging task.
For a start, let us investigate the static case with $f_s$ being the static
transfer function of a gate.  In this setup, an output derivative $\vout'=0$ is
achieved for all values fulfilling the condition $\vout = f_s(\vin)$ since $f_s$
represents the stable states of the gate.  To obtain high accuracy when
discretizing the analog signal, one typically chooses the output threshold 
such that the respective output waveform for a full-range input pulse has a steep
slope at this point.  While $\vthoutm = \vdd /2$ is in general a good choice, the
corresponding $\vthinm$ will differ significantly between balanced and
high-threshold inverters, for example.

Besides these static considerations, for a dynamic input, coupling
capacitances cause a current at the output, which must be compensated via the
gate-source voltages of the transistors as well.  Obviously, the required
overshoot w.r.t.\ $\vthin$, and hence the time until this value is reached,
depends on many parameters like the size of the coupling capacitances and the
slope of the input signal.

It is also worth mentioning that CMOS logic shows an inverse proportionality
between in- and output, i.e., increasing the input increases the conductivity to
\gnd, which drains the output, and vice versa.  This implies that the current
induced by the alternation of the input goes against the change of the output
current, which effectively slows down the switching operation, and thus
increases the pure delay $\dmin>0$.

\cref{lem:vth:inout} has a severe consequence for the simulation of circuits
in any model, like IDM, where \cref{lem:delta:min:idm} holds:

\begin{observation}\label{obs:singlefixesall}
  Fixing either $\vthin$ or $\vthout$ for a single gate $G$ fixes the threshold
  voltages of all gates in the circuit when it is simulated in a model where
  \cref{lem:vth:inout} holds.
\end{observation}

\begin{proof}[Justification]
  For a single gate, the observation follows from the unique relationship
  between $\vthinm$ and $\vthoutm$ by \cref{lem:vth:inout}. For
  gates in a path, the ability to just add up their delays requires
  output and next input discretization thresholds to be the same.
\end{proof}

Since the detailed relation of $\vthinm$ and $\vthoutm$ according to
  \cref{lem:vth:inout} depends on the individual gate, this means that the
  discretization threshold voltages across a circuit may vary in \emph{a priori} arbitrary
  ways, depending on the interconnect topology and the gate properties.
In the extreme, the discretization thresholds can exceed the acceptable range
  for some gates, and even lead to unsatisfiable assignments in the case of
  circuits containing feedback-loops.
In any case, it may take a large effort to properly characterize
  every gate such that the dependencies among discretization thresholds are fulfilled.
By contrast, an ideally composable delay model uses a uniform discretization threshold such as
  $\vthout=\vthin = \vdd /2$.

To investigate if \gls{idm} allows such a uniform choice, we proceed with \cref{lem:vth:non:match}:

\begin{observation} \label{lem:vth:non:match} Characterizing a gate with
  non-matching discretization thresholds $\vthin$ and $\vthoutm$, where matching
  $\vthinm$ and $\vthoutm$ lead to an IDM channel with pure delay $\bdmin$,
  results in delay functions $\dup(T)$, $\ddo(T)$, which satisfy
  $\dup(-\dupmin) = \dupmin$ and $\ddo(-\ddomin) = \ddomin$ for
  $\dupmin = \bdmin + \Delta^+ \neq \ddomin=\bdmin + \Delta^-$.  $\Delta^+$ and
  $\Delta^-$ have opposite sign, with $\Delta^+>0$ for $\vthin < \vthinm$.
\end{observation}

\begin{proof}[Justification]
  The observation follows from refining the argument used for confirming
  \cref{lem:vth:inout}, where it was shown how matching $\vthinm$ and $\vthoutm$
  are achieved. For the non-matching case, we increase resp.\ decrease $\vthin$,
  starting from $\vthinm$, while keeping everything else, i.e., electronic
  characteristics, waveforms and $\vthout$, unchanged. As illustrated in
  \cref{fig:im_char} for $\vthin<\vthinm$, it still takes $\bdmin$ from hitting
  $\vthinm$ (at time $t_o-\bdmin$) to seeing a zero time glitch (at time $t_o$)
  at the output. The falling transition has already crossed $\vthinm$ when it
  hits on $\vthin$, whereas the rising transition still has some way to go:
  Denoting the switching waveforms of the preceding gate (driving the input) by
  $\fup$ and $\fdo$, the pure delay for the rising resp.\ falling transition
  evaluates to $\dupmin= \bdmin+\Delta^+$ and $\ddomin = \bdmin+\Delta^-$ with
\begin{equation}
\Delta^+ =  \fup^{-1}(\vthinm) - \fup^{-1}(\vthin),\ \ \ 
\Delta^- =  \fdo^{-1}(\vthinm) - \fdo^{-1}(\vthin)\label{eq:fdupmin}.
\end{equation}
Consequently, $\dup(-\dupmin) = \dupmin$ and $\ddo(-\ddomin) = \ddomin$ indeed
holds.  Finally, since $\fup$ must obviously rise and $\fdo$ must fall, it
follows that if $\Delta^+>0$ (the case in \cref{fig:im_char}) then
$\Delta^- <0$.
\end{proof}

\begin{figure}[tb]
  \centering
  \includegraphics[width=0.7\columnwidth]{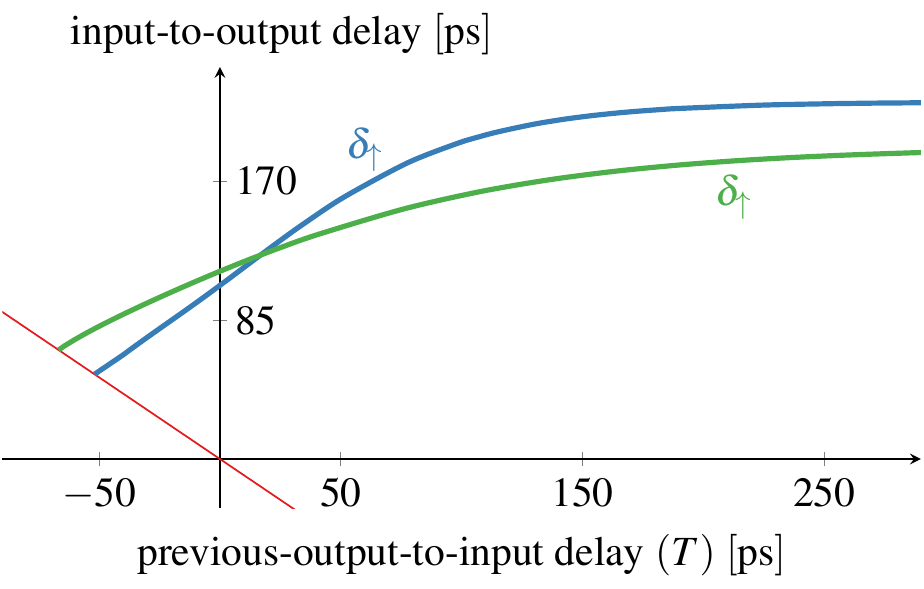}
  \caption{Characterizing a gate with $\vthin=\vthout=\vdd/2$.}
  \label{fig:inverter_gidm}
\end{figure}

\cref{fig:inverter_gidm} shows the derived delay function for non-matching
discretization thresholds. Clearly visible are the different pure delays $\dupmin\neq \ddomin$. Please note, that in
our justification of \cref{lem:vth:non:match}, we focused on $\dmin$ and how it
changes with varying discretization threshold voltages. If the rising and falling switching
waveforms were always the same, as is assumed in the analog channel model
\cref{fig:analogy}, this would result in delay functions that are fixed in shape
and are simply shifted along the 2\textsuperscript{nd} median.  The actual delay
functions of gates, obtained by analog simulations, for example, exhibit
additional deviations (for $T \neq \dudmin$), however, since the shape of the
input switching waveforms also vary. Consequently, the difference between
$\vthinm$ and $\vthin$ will not always be passed in constant time.

%
%
Finally, the dependency of the IDM on the particular choice of the 
  discretization threshold voltages also reveals a different problem:

\begin{observation}\label{lem:noosc}
  Different choices of $\vthout$ can significantly change the digital model prediction of the IDM.
\end{observation}

\begin{proof}[Justification]
  Sub-threshold pulses are automatically removed by the comparator in
  \cref{fig:analogy}, i.e., they are completely invisible in the digital signal
  that is fed to the successor gate. Consequently analog traces, that do not
  cross $\vthout$, e.g., high-frequency oscillations at intermediate voltage
  levels, may be totally suppressed: Assume an oscillatory behavior of a gate
  output with minimal voltage $V_0$ and maximal $V_1$. These oscillations would
  only be reflected in the digital discretization if $\vthout \in (V_0, V_1)$.
\end{proof}

\section{Composable Involution Delays}
\label{sec:gidm}

In this section, we define our \emph{Composable Involution Delay Model}
(CIDM), which allows to circumvent the problems presented in the previous
section: According to \cref{lem:vth:non:match}, using non-matching thresholds introduces a
pure delay shift. The major building blocks of our CIDM are hence \emph{PI channels},
which consist of a pure delay shifter with different shifts $\Delta^+$ and
$\Delta^-$ for rising and falling transitions
\footnote{A pure delay shifter with
  $\Delta^+\neq\Delta^-$ causes a constant extension/compression of up/down
  input pulses by $\pm(\Delta^+ - \Delta^-)$.}
followed by an \gls{idm} channel.
In order to also alleviate the problem of invisible oscillations identified in
\cref{lem:noosc}, we re-shuffle the internal architecture of the
original involution channels shown in \cref{fig:analogy} in order to expose trains of
canceled transitions on the interconnecting wires.

\begin{theorem}[PI channel properties]\label{thm:pcchannel}
  Consider a channel PI formed by the concatenation of a pure delay shifter
  $(\Delta^+,\Delta^-)$ with $\Delta^+,\Delta^- \in \mathbb{R}$ followed by an
  involution channel $c$, given via $\bdup(.)$ and $\bddo(.)$ with minimum
  delay $\bdmin$. Then
  PI is \emph{not} an involution channel, but rather characterized by delay
  functions defined as
\begin{equation}
\dup(T)=\Delta^++\bdup(T+\Delta^+) \qquad
\ddo(T)=\Delta^-+\bddo(T+\Delta^-).\label{eq:defpcup}
\end{equation}
These functions satisfy 
\begin{align}
\dup\bigl(-\ddo(T) - (\Delta^+ - \Delta^-)\bigr) &= -T + (\Delta^+ - \Delta^-)\label{eq:dupddopc}\\
\ddo\bigl(-\dup(T) + (\Delta^+ - \Delta^-)\bigr) &= -T - (\Delta^+ - \Delta^-)\label{eq:ddoduppc}\\
\dup(-\dupmin) &= \dupmin\label{eq:dupmindef}\\
\ddo(-\ddomin) &= \ddomin\label{eq:ddomindef}
\end{align}
for $\dupmin = \bdmin  + \Delta^+$ and  $\ddomin = \bdmin  + \Delta^-$.
\end{theorem}

\begin{proof}
  Consider an input signal consisting of a simple negative pulse, as depicted in
  \cref{fig:single_history}.  Let $t_i'$ resp. $t_i$ be the time of the falling
  resp.\ rising input transition, $t_p'$ resp.  $t_p$ the time of the falling
  resp.\ rising transition at the output of the pure delay shifter, and $t_o'$
  resp. $t_o$ the time of the falling resp.\ rising transition after the
  involution channel.  With $\bT=t_p-t_o'$, we get $\bdup(\bT)=t_o-t_p$ as well
  as $t_p=t_i+\Delta^+$ and $t_p'=t_i'+\Delta^-$.

For the delay function $\dup(T)$ of the PI channel, if we 
set $T=t_i-t_o' = t_i -t_p + t_p -t_o' = -\Delta^+ + \bT$, we
find 
\begin{align}
\dup(T)&=t_o-t_i = t_o - t_p + t_p - t_i = \Delta^++\bdup(\bT)\nonumber\\
&= \Delta^++\bdup(T+\Delta^+)\label{eq:gduppc}
\end{align}
as asserted. By setting $T=-\bdmin - \Delta^+$ and using $\bdup(-\bdmin)=\bdmin$
the equality $\dup(-\bdmin - \Delta^+) = \Delta^+ + \bdmin$ is achieved, which
confirms \cref{eq:dupmindef}.

By analogous reasoning for an up-pulse at the input, which results in the same
equations as above with $\Delta^+$ exchanged
with $\Delta^-$ and $\bdup(\bT)$ with $\bddo(\bT)$, we also get
\begin{align}
\ddo(T)&=t_o-t_i = t_o - t_p + t_p - t_i = \Delta^-+\bddo(\bT)\nonumber\\
&= \Delta^-+\bddo(T+\Delta^-)\label{eq:gddopc}
\end{align}
as asserted. Setting $T=-\bdmin - \Delta^-$ and using $\bddo(-\bdmin)=\bdmin$
confirms \cref{eq:ddomindef} as well.

Using a simple parameter substitution allows to transform \cref{eq:gduppc} and
\cref{eq:gddopc} to
\begin{align}
\bdup(\bT)  &= \dup(\bT-\Delta^+) - \Delta^+  \\
\bddo(\bT) &=  \ddo(\bT-\Delta^-) - \Delta^-.
\end{align}
Utilizing these in the involution property of~$\bdup$ and~$\bddo$ provides
\begin{align}
\bT &= -\bdup\bigl(-\bddo(\bT)\bigr)\nonumber\\
&= -\dup\bigl(-\bddo(\bT)-\Delta^+\bigr) + \Delta^+\nonumber\\
&= -\dup\bigl(-\bigl(\ddo(\bT-\Delta^-)-\Delta^-\bigr)-\Delta^+\bigr) + \Delta^+\nonumber\\
&= -\dup\bigl(-\ddo(\bT-\Delta^-) + \Delta^- - \Delta^+\bigr) + \Delta^+\nonumber.
\end{align}
If we substitute $T = \bT - \Delta^-$ in the last line, we arrive at
\begin{equation}
T - (\Delta^+ - \Delta^-) = -\dup\bigl(-\ddo(T) - (\Delta^+ - \Delta^-)\bigr),
\end{equation}
which confirms \cref{eq:dupddopc}.  

Doing the same for the reversed involution
property of $c$, provides
\begin{align}
\bT &= -\bddo\bigl(-\bdup(\bT)\bigr) \nonumber\\
&= -\ddo\bigl(-\bdup(\bT)-\Delta^-\bigr) + \Delta^-\nonumber\\
&= -\ddo\bigl(-\bigl(\dup(\bT-\Delta^+)-\Delta^+\bigr)-\Delta^-\bigr) + \Delta^-\nonumber\\
&= -\ddo\bigl(-\dup(\bT-\Delta^+) + \Delta^+ - \Delta^-\bigr) + \Delta^-\nonumber.
\end{align}
If we substitute $T = \bT - \Delta^+$ in the last line, we arrive at
\begin{equation}
T + (\Delta^+ - \Delta^-) = -\ddo\bigl(-\dup(T) + \Delta^+ - \Delta^-\bigr),
\end{equation}
which confirms \cref{eq:ddoduppc}.

\end{proof}

Eq.~\cref{eq:defpcup} implies that $\dup(.)$ resp.\
$\ddo(.)$ are the result of shifting $\bdup(.)$ resp.\ $\bddo(.)$ along the
2\textsuperscript{nd} median by $\Delta^+$ resp.\ $\Delta^-$.  
It is apparent from \cref{fig:inverter_gidm}, though, that the choice of $\Delta^+$, $\Delta^-$ cannot be arbitrary, as it restricts the range of feasible values for $T$ via the
domain of $\bdup(.)$ resp.\ $\bddo(.)$ (see \cref{def:compatibility} for
further details).


  This becomes even more apparent in the analog channel model.
  \cref{fig:channel_model_idm2gidm}~(a) shows an extended block diagram of an
  IDM channel where we applied two changes: First, we added a (one-input,
  one-output) zero-time Boolean gate~$G$.  Second, we split the comparator at
  the end into a thresholder~$\thresholder$ and a cancellation unit~$C$. The
  \emph{thresholder unit}~$\thresholder$ outputs, for each transition on $u_d$,
  a corresponding $\vth$-crossing time of $u_r$, independently of whether it
  will actually be reached or not.  For subthreshold pulses, the transition
  might even be scheduled in the past. The \emph{cancellation unit}~$C$ only
  propagates transitions that are in the correct temporal order. The
  components~$\thresholder$ and~$C$ together implement the same functionality as
  a comparator.

  At the beginning of the channel, the Boolean gate $G$ (we assume a
  single-input gate for now) evaluates the input signal $u_i$ in zero time and
  outputs $u_g$, which is subsequently delayed by the pure delay shifter
  $\Deltapm$. Here lies the cause of the problem: If either $\Delta^+<0$ or
  $\Delta^-<0$ it is possible that transitions $u_p$ are in reversed temporal
  order which, after being delayed by the constant pure delay $\bdmin$, have to
  be processed in this fashion by the slope delimiter. The latter is, however,
  only defined on traces encoded via the alternating Boolean signal transitions'
  \emph{\gls{wst}}, which occur in a strictly increasing temporal order and mark the
  points in time when the switching waveforms shall be changed. Placing the delay
  shifter in front of the gate, as shown in
  \cref{fig:channel_model_idm2gidm}~(b), does not change the situation, since
  the gate also expects transitions in the correct temporal order (note that
  this is not equal to \gls{wst} since the pure delay is still missing).

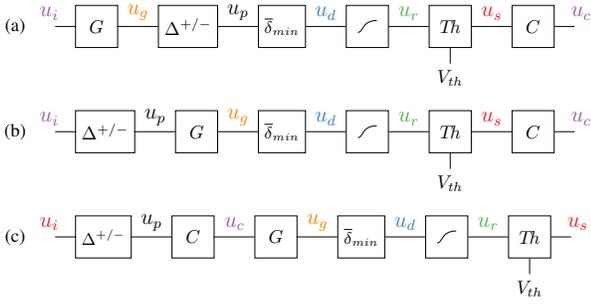
\begin{figure}[tb]
  \centering
  \scalebox{0.7}{\begin{tikzpicture}[
xscale=0.8,
yscale=0.5,
>=latex'
]

\def\figcirampin{38}
\def\figcirampinlen{0.25}
\def\figcirampvoltlen{0.27}
\def\figcirampcap{0.3}
\def\figcirdista{0.95}
\def\figcirdistb{1.25}
\def\figcirdistc{1}
\def\figciroutlen{0.25}
\def\figmidrow{-1.8}
\def\figbotrow{-4}


\node at (-2.5*\figcirdista,0) {(a)};

\node [anchor=west, draw, minimum size=0.8cm] (G) at (-\figcirdista,0) {$G$};
\path [draw, shorten <=-0.5pt]
(G.west) -- node [anchor=south east] {\color{color4} \large $u_i$} ++(-\figcirdista/2,0);

\path [draw, shorten <=-0.5pt]
(G.east) -- node[above] {\color{color5} \large $u_g$} ++(\figcirdista,0) node (dshift) [draw, minimum size=0.8cm, anchor=west] {$\Deltapm$};

\path [draw, shorten <=-0.5pt]
(dshift.east) -- node[above] {\color{black} \large $u_p$} ++(\figcirdista,0) node (dmin) [draw, minimum size=0.8cm, anchor=west] {\small $\bdmin$};

\path [draw]
(dmin.east) -- node [anchor=south] {\color{color2} \large  $u_d$} ++(\figcirdista,0) node (srl) [draw, minimum size=0.8cm, anchor=west] {};
\path [draw, scale=0.45]
(srl) ++(-0.5,-0.5) .. controls +(0.3,0) and +(-0.8,0) .. +(1,1);

\path [draw]
(srl.east) -- node [anchor=south] {\color{color3} \large  $u_r$}++(\figcirdista,0) node (S) [draw, minimum size=0.8cm, anchor=west] {$T\!h$};
\path [draw] (S.south) -- ++(0, -\figcirdistc*0.6) node[anchor=north] {$\vth$};

\path [draw, shorten <=-0.5pt]
(S.east) -- node[above, color=color1] {\large $u_s$} ++(\figcirdista,0) node (C) [anchor=west, draw, minimum size=0.8cm] {$C$};
\path [draw] (C.east) -- node[anchor=south west] {\color{color4} \large ${u_c}$} ++(\figcirdista/2,0);


\begin{scope}[yshift=-4cm]

\node at (-2.5*\figcirdista,0) {(b)};

\node [anchor=west, draw, minimum size=0.8cm] (dshift) at (-\figcirdista,0) {$\Deltapm$};
\path [draw, shorten <=-0.5pt]
(dshift.west) -- node [anchor=south east] {\color{color4} \large $u_i$} ++(-\figcirdista/2,0);

\path [draw, shorten <=-0.5pt]
(dshift.east) -- node[above] {\color{black} \large $u_p$} ++(\figcirdista,0) node (G) [draw, minimum size=0.8cm, anchor=west] {$G$};

\path [draw, shorten <=-0.5pt]
(G.east) -- node[above] {\color{color5} \large $u_g$} ++(\figcirdista,0) node (dmin) [draw, minimum size=0.8cm, anchor=west] {\small $\bdmin$};

\path [draw]
(dmin.east) -- node [anchor=south] {\color{color2} \large  $u_d$} ++(\figcirdista,0) node (srl) [draw, minimum size=0.8cm, anchor=west] {};
\path [draw, scale=0.45]
(srl) ++(-0.5,-0.5) .. controls +(0.3,0) and +(-0.8,0) .. +(1,1);

\path [draw]
(srl.east) -- node [anchor=south] {\color{color3} \large  $u_r$}++(\figcirdista,0) node (S) [draw, minimum size=0.8cm, anchor=west] {$T\!h$};
\path [draw] (S.south) -- ++(0, -\figcirdistc*0.6) node[anchor=north] {$\vth$};

\path [draw, shorten <=-0.5pt]
(S.east) -- node[above, color=color1] {\large $u_s$} ++(\figcirdista,0) node (C) [anchor=west, draw, minimum size=0.8cm] {$C$};
\path [draw] (C.east) -- node[anchor=south west] {\color{color4} \large ${u_c}$} ++(\figcirdista/2,0);

\end{scope}


\begin{scope}[yshift=-8cm]

\node at (-2.5*\figcirdista,0) {(c)};

\node [anchor=west, draw, minimum size=0.8cm] (dshift) at (-\figcirdista,0) {\small $\Deltapm$};
\path [draw, shorten <=-0.5pt]
(dshift.west) -- node [anchor=south east] {\color{color1} \large $u_i$} ++(-\figcirdista/2,0);

\path [draw, shorten <=-0.5pt]
(dshift.east) -- node[above, color=black] {\large $u_p$} ++(\figcirdista,0) node (C) [anchor=west, draw, minimum size=0.8cm] {$C$};

\path [draw, shorten <=-0.5pt]
(C.east) -- node[above] {\color{color4} \large $u_c$} ++(\figcirdista,0) node (G) [draw, minimum size=0.8cm, anchor=west] {$G$};

\path [draw, shorten <=-0.5pt]
(G.east) -- node[above] {\color{color5} \large $u_g$} ++(\figcirdista,0) node (dmin) [draw, minimum size=0.8cm, anchor=west] {\small $\bdmin$};

\path [draw]
(dmin.east) -- node [anchor=south] {\color{color2} \large  $u_d$} ++(\figcirdista,0) node (srl) [draw, minimum size=0.8cm, anchor=west] {};
\path [draw, scale=0.45]
(srl) ++(-0.5,-0.5) .. controls +(0.3,0) and +(-0.8,0) .. +(1,1);

\path [draw]
(srl.east) -- node [anchor=south] {\color{color3} \large  $u_r$}++(\figcirdista,0) node (S) [draw, minimum size=0.8cm, anchor=west] {$T\!h$};
\path [draw] (S.south) -- ++(0, -\figcirdistc*0.6) node[anchor=north] {$\vth$};
\path [draw] (S.east) -- node[anchor=south west] {\color{color1} \large ${u_s}$} ++(\figcirdista/2,0);

\end{scope}

\end{tikzpicture}}
  \caption{Candidate channel models for \gls{cidm}.}
  \label{fig:channel_model_idm2gidm}
\end{figure}

One possibility to avoid transitions in the wrong temporal order at the Boolean
gate is to move the canceling unit to the front, as shown in
\cref{fig:channel_model_idm2gidm}~(c). This solves our present problems but has
the consequence, that now transitions are interchanged among gates using the
\emph{\gls{tct}} encoding: The \gls{tct} encoding gives, in sequence order, the points
in time when the analog switching waveform would have crossed $\vth$ (it is not
required that it actually does). Consequently, a signal given in \gls{tct} also
exposes canceled transitions. Actually this is very convenient for us, since
this allows us implicitly to detect oscillations independent of the chosen
output threshold and thus solves design challenge~(3) from the introduction.

Not all signals in \cref{fig:channel_model_idm2gidm} can
actually be mapped to \gls{tct} or \gls{wst}; by suitably recombining the
components in our CIDM channel, however, these encodings will be sufficient for
our purposes.  More specifically, \gls{tct} will be created by the scheduler
$S$, subsequently modified by the delay shifter, altered by the cancellation
unit $C$, evaluated by the Boolean gate and finally transformed to \gls{wst} by
$\bdmin$.

Now we are finally ready to formally define a \gls{cidm} channel (see
\cref{fig:channel_model_gidm}).
Note that, although a PI channel
differs by its internal structure significantly from the \gls{cidm} channel, 
they are equivalent with respect to \cref{thm:pcchannel}. 

\begin{definition}
  A \gls{cidm} channel comprises in succession of a pure delay shifter, a
  cancellation unit, a Boolean gate, a pure-delay unit, a shaping unit and a
  thresholding unit [see \cref{fig:channel_model_idm2gidm}(c)].
\end{definition}

One may
wonder whether CIDM channels could be partitioned also in a different
fashion. The answer is yes, several other partitions are possible. For example,
one could transmit signal $u_g$ and move the slew-rate limiter and the scheduler
to the succeeding channel. This would, however, mean that properties of single
CIDM channels depend on the properties of both predecessor and successor gate,
which complicates the process of channel characterization
and parametrization based on design parameters.

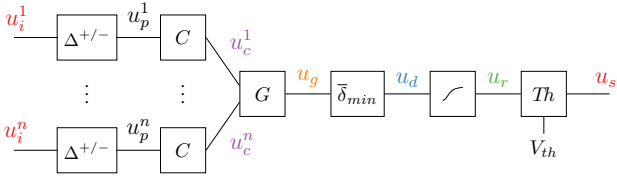
\begin{figure}[tb]
  \centering
  \scalebox{0.75}{\begin{tikzpicture}[
xscale=0.8,
yscale=0.5,
>=latex'
]

\def\figcirampin{38}
\def\figcirampinlen{0.25}
\def\figcirampvoltlen{0.27}
\def\figcirampcap{0.3}
\def\figcirdista{0.95}
\def\figcirdistb{1.25}
\def\figcirdistc{1}
\def\figciroutlen{0.25}
\def\figmidrow{-1.8}
\def\figbotrow{-4}

\node[align=center, draw, minimum size=0.8cm] (del) at (0,0) {\small $\Deltapm$};
\path [draw, shorten <=-0.5pt]
(del.west) -- node [anchor=south east] {\color{color1} \large ${u_i^1}$}
++(-\figcirdista,0);

\path [draw, shorten <=-0.5pt]
(del.east) -- node[above, color=black] {\large $u_p^1$} ++(\figcirdista,0) node (IF) [anchor=west, draw, minimum size=0.8cm] {$C$};


\node[align=center, draw, minimum size=0.8cm] (deln) at (0,\figbotrow) {\small $\Deltapm$};
\path [draw, shorten <=-0.5pt]
(deln.west) -- node [anchor=south east] {\color{color1} \large ${u_i^n}$}
++(-\figcirdista,0);

\path [draw, shorten <=-0.5pt]
(deln.east) -- node[above, color=black] {\large $u_p^n$} ++(\figcirdista,0) node (IFn) [anchor=west, draw, minimum size=0.8cm] {$C$};


\node at ($(IF)!0.45!(IFn)$) {$\vdots$};
\node at ($(del)!0.45!(deln)$) {$\vdots$};


\node[draw, minimum size=0.8cm, anchor=west] (fb) at ($(IF)!0.5!(IFn) + (\figcirdistb,0)$) {$G$};
\path [draw]
(IF.east) -- node [anchor=south west] {\color{color4} \large  $u_c^1$} (fb.160);
\path [draw]
(IFn.east) -- node [anchor=north west] {\color{color4} \large  $u_c^n$} (fb.200);

\path [draw]
(fb.east) -- node [anchor=south] {\color{color5} \large  $u_g$}
++(\figcirdistc,0) node (dmin) [draw, minimum size=0.8cm, anchor=west] {$\bdmin$};

\path [draw]
(dmin.east) -- node [anchor=south] {\color{color2} \large  $u_d$}
++(\figcirdistc,0) node (srl) [draw, minimum size=0.8cm, anchor=west] {};
\path [draw, scale=0.45]
(srl) ++(-0.5,-0.5) .. controls +(0.3,0) and +(-0.8,0) .. +(1,1);

\path [draw]
(srl.east) -- node [anchor=south] {\color{color3} \large  $u_r$}
++(\figcirdistc,0) node (S) [draw, minimum size=0.8cm, anchor=west] {$T\!h$};
\path [draw] (S.south) -- ++(0, -\figcirdistc*0.6) node[anchor=north] {$\vth$};
\path [draw] (S.east) -- node[anchor=south west] {\color{color1} \large ${u_s}$} ++(\figcirdista,0);

\end{tikzpicture}}
  \caption{Channel model for \gls{cidm}.}
  \label{fig:channel_model_gidm}
\end{figure}

The main practical advantage of a \gls{cidm} channel, which is a generalization of 
an \gls{idm} channel (just set $\Delta^-=\Delta^+=0$), is the additional
degree of freedom for gate characterization in conjunction with the fact that a single
channel encapsulates a single gate.

\section{Glitch Propagation in the CIDM}
\label{sec:gidmisidm}

Since CIDM channels do not satisfy the involution property, the question about
faithful glitch propagation arises. After all, the proof of faithfulness of IDM
\cite{FNNS19:TCAD} rests on the continuity of IDM channels, which has been
proved only for involution delay functions. In this section, we will show that,
for every modeling of a circuit with our CIDM channels, there is an equivalent
modeling with IDM channels. Consequently, faithfulness of the IDM carries over
to the CIDM.

For this purpose, we consider two successive \gls{cidm} channels and
investigate the \emph{logical channel}, i.e., the interconnection between two
gates $G_1$ and $G_2$ as shown in \cref{fig:channel_model_proofs}. For
conciseness, we integrate $\bdmin$, the slew-rate limiter and the
scheduler $S$ in a new block $DSS$, and $\Deltapm$ followed by~$\thresholder$ in
the new block $P\thresholder$.
Using this notation, the logical channel consists of the $DSS$ block of the
predecessor gate $G_1$ and the $P\thresholder$ block of the successor gate
$G_2$. Overall, this is just an IDM channel followed by a pure
delay shifter, which will be denoted in the sequel as \emph{IP channel}. The
following \cref{thm:cpchannel} proves the somewhat surprising fact that every IP
channel satisfies the properties of an involution channel:

\begin{theorem}[IP channel properties]\label{thm:cpchannel}
  Consider an IP channel formed by an involution channel given via $\bdup(.)$,
  $\bddo(.)$, followed by a pure delay shifter $(\Delta^+,\Delta^-)$ with
  $\Delta^+,\Delta^- \in \mathbb{R}$. Then, it is an involution channel,
  characterized by some delay functions $\dup(.)$, $\ddo(.)$.
\end{theorem}

\begin{proof}
Consider an input signal consisting of a single negative pulse, as depicted in \cref{fig:single_history}.
Let $t_i'$ resp. $t_i$ be the time of the falling resp.\ rising input
transition, $t_c'$ resp. $t_c$ the time of the falling resp.\ rising
transition at the output of the involution channel, and $t_o'$ resp. $t_o$ 
the time of the falling resp.\ rising transition after the pure delay
shifter. With $\bT=t_i-t_c'$, we get $\bdup(\bT)=t_c-t_i$ as well
as $t_o'=t_c'+\Delta^-$ and $t_o=t_c+\Delta^+$.

By setting $T=t_i-t_o' = t_i -t_c' + t_c' -t_o' = \bT -\Delta^-$ for the delay
function $\dup(T)$ of the IP channel we find
\begin{align}
\dup(T)&=t_o-t_i = t_o - t_c + t_c - t_i = \Delta^++\bdup(\bT)\nonumber\\
&= \Delta^++\bdup(T+\Delta^-).\label{eq:gdup}
\end{align}

By analogous reasoning for a an up-pulse at the input, which results in
the same equations as above with $\Delta^+$ exchanged
with $\Delta^-$ and $\bdup(\bT)$ with $\bddo(\bT)$, we also get
\begin{align}
\ddo(T)&=t_o-t_i = t_o - t_c + t_c - t_i = \Delta^-+\bddo(\bT)\nonumber\\
&= \Delta^-+\bddo(T+\Delta^+).\label{eq:gddo}
\end{align}

Equations \eqref{eq:gdup} and \eqref{eq:gddo} are equivalent to
\begin{align}
\bdup(\bT)  &= \dup(\bT-\Delta^-) - \Delta^+  \\
\bddo(\bT) &=  \ddo(\bT-\Delta^+) - \Delta^-
\end{align}
which can be used in the involution property of~$\bdup$ and~$\bddo$ to achieve
\begin{align}
\bT &= -\bdup\bigl(-\bddo(\bT)\bigr) \nonumber\\
&= -\dup\bigl(-\bddo(\bT)-\Delta^-\bigr) + \Delta^+\nonumber\\
&= -\dup\bigl(-\bigl(\ddo(\bT-\Delta^+)-\Delta^-\bigr)-\Delta^-\bigr) + \Delta^+\nonumber\\
&= -\dup\bigl(-\ddo(\bT-\Delta^+)\bigr) + \Delta^+\label{equ:surprise}
\end{align}
which confirms that the IP channel is indeed an involution channel.
\end{proof}

We note that $\rdmin$ of the IP channel is usually different from $\bdmin$ of
the constituent IDM channel.
Indeed, \cref{eq:gdup} above shows that $\rdmin$ is defined by
$\rdmin - \Delta^+ = \bdup(-\rdmin + \Delta^-)$, for example, which
reveals that we may (but need not) have $\rdmin \neq \bdmin$. In addition, the IP channel is
strictly causal only if $\Delta^+$, $\Delta^-$ satisfy certain conditions: From
\cref{eq:gdup} and \cref{eq:gddo} and the required conditions
$\dup(0)>0 \Leftrightarrow \ddo(0) >0$, we get
\begin{equation}
\dup(0) = \Delta^+ + \bdup(\Delta^-) > 0 \Leftrightarrow \ddo(0) = \Delta^- + \bddo(\Delta^+) > 0.\label{eq:gidmcausality}
\end{equation} 

At this point, the question arises whether it can be ensured that the logical
channels in \cref{fig:channel_model_proofs} are strictly causal. The answer is
yes, provided the interconnected gates are \emph{compatible}, in the sense that
the joined $P\thresholder$ block of $G_2$ and the $DSS$ block of $G_1$ are
compatible w.r.t.\ \cref{lem:vth:non:match}.  More specifically, the pure delays
$\Delta^+$, $\Delta^-$ embedded in $P\thresholder$ of $G_2$ should ideally match
(the switching waveforms of) the $DSS$ block in $G_1$: According to
\cref{eq:fdupmin}, $-\Delta^+$ is the time the rising input waveform of $G_2$
requires to change from $\vthinm$ to the actual threshold voltage $\vthin$,
while $-\Delta^-$ denotes the same for the falling input. Assuming $\Delta^+>0$
and $\Delta^-<0$, it can be seen from \cref{eq:gddo} that the overall delay
function for falling transitions is derived by shifting the original one to the
left and downwards (cp.\ \cref{fig:inverter_gidm}). Note that the 2\textsuperscript{nd} median is crossed at the
same location since $\bddo(\bdmin+\Delta^+) = \bdmin - \Delta^-$, and thus results
in $\bdmin=\rdmin>0$ (which implies causality). The case of $\Delta^+<0$ and
$\Delta^->0$ can be argued analogously, starting from \eqref{eq:gdup}. 

These considerations justify the following definition:

\begin{figure}[tb]
  \centering
  \scalebox{0.75}{\begin{tikzpicture}[
xscale=0.8,
yscale=0.5,
>=latex'
]

\def\figcirampin{38}
\def\figcirampinlen{0.25}
\def\figcirampvoltlen{0.27}
\def\figcirampcap{0.3}
\def\figcirdista{0.95}
\def\figcirdistb{1.25}
\def\figcirdistc{1}
\def\figciroutlen{0.25}
\def\figmidrow{-1.8}
\def\figbotrow{-4}

\node [anchor=west, draw, minimum size=0.8cm] (DC) at (0,0) {$P\thresholder$};
\path [draw, shorten <=-0.5pt]
(DC.west) -- node [anchor=south east] {\color{color3} \large $u_i$} ++(-\figcirdista/2,0);

\path [draw, shorten <=-0.5pt]
(DC.east) -- node[above] {\color{color2} \large $u_c$} ++(\figcirdista,0) node (G) [draw, minimum size=0.8cm, anchor=west] {$G_1$};

\path [draw]
(G.east) -- node [anchor=south] {\color{color2} \large  $u_g$}++(\figcirdista,0) node (S) [draw, minimum size=0.8cm, anchor=west] {$DSS$};
\path [draw] (S.south) -- ++(0, -\figcirdistc*0.6) node[anchor=north] {$\vth$};

\path [draw, shorten <=-0.5pt]
(S.east) -- node[anchor=south, pos=0.2] {\color{color3} \large ${u_s}$} node[above, pos=0.8] {\color{color3} \large $u_i$} ++(2*\figcirdista,0) node (DC) [draw, minimum size=0.8cm, anchor=west] {$P\thresholder$};

\coordinate (TMP) at ($(S)!0.5!(DC)$);
 \draw[dashed] (  TMP |- 0,-2) -- (TMP |- 0,3);

\path [draw, shorten <=-0.5pt]
(DC.east) -- node[above] {\color{color2} \large $u_c$} ++(\figcirdista,0) node (G) [draw, minimum size=0.8cm, anchor=west] {$G_2$};

	\draw[line width=1.6pt, color=color5] ($(-0.1,-0.2) +(S.south west)$) rectangle ($(0.1,0.2) + (DC.north east)$);
	\node[color=color5, fill=white, inner sep=0] at (TMP |- 0,1.5) {logical channel};

\path [draw]
(G.east) -- node [anchor=south] {\color{color2} \large  $u_g$}++(\figcirdista,0) node (S) [draw, minimum size=0.8cm, anchor=west] {$DSS$};
\path [draw] (S.south) -- ++(0, -\figcirdistc*0.6) node[anchor=north] {$\overline{\vth}$};
\path [draw] (S.east) -- node[anchor=south west] {\color{color3} \large ${u_s}$} ++(\figcirdista/2,0);

	\node at (3*\figcirdista,2.5) {Gate A};
	\node at (10*\figcirdista,2.5) {Gate B};

\end{tikzpicture}}
  \caption{Channel model for proofs of the \gls{cidm}. Signals in blue
    have data type WST, those in green VCT.}
  \label{fig:channel_model_proofs}
\end{figure}
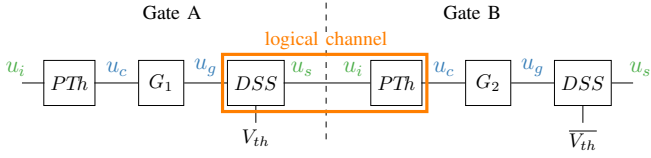

\begin{definition}[Compatibility of CIDM channels]\label{def:compatibility}
Two interconnected CIDM channels are called \emph{compatible}, if the
logical channel between them is strictly causal.
\end{definition}

Consequently, a logical channel connecting $G_1$ and $G_2$ is strictly causal if
$\Delta^+$, $\Delta^-$ has been determined in accordance with
\cref{lem:vth:non:match}. If this is not observed, non-causal effects like an
output pulse crossing $\vthout$ without the corresponding input pulse crossing
$\vthin$ could appear.

Every chain of gates properly modeled in the CIDM can be
represented by a chain of Boolean gates interconnected by causal IDM channels,
with a ``dangling'' $P\thresholder$ block at the very beginning and a $DSS$ block at the
very end. Whereas the latter is just an IDM channel, 
this is not the case for the former. Fortunately, this does
not endanger the applicability of the existing IDM results either: As stated in
property C2) of a circuit in \cite[Sec.~III]{FNNS19:TCAD}, the original IDM
assumes 0-delay channels for connecting an output port of a predecessor circuit
$C_1$ with an input port of a successor circuit $C_2$. In the case of using CIDM
for modeling $C_1$ and $C_2$, this amounts to combining the $DSS$ block of gate
that drives the output port of $C_1$ with the $P\thresholder$ block of the input of the
gate in $C_2$ that is attached to the input port.  Note that the analogous
reasoning also applies to any feedback loop in a circuit.

On the other hand, for an ``outermost'' input port of a circuit, we can
just require that the connected gate must have a threshold voltage 
matching the external input signal, such that $\dupmin=\ddomin=0$ 
for the dangling $P\thresholder$ component. Finally, in hierarchical simulations,
where the output ports of some circuit are connected to the input ports
of some other circuits, the situation explained in the previous paragraph
reappears.

As a consequence, all the results and all the machinery developed for the
original IDM~\cite{FNNS19:TCAD} could in principle be applied also to
circuits modeled with CIDM channels. Both impossibility results and possibility
results, and hence faithfulness, hold, and even the IDM digital timing
simulation algorithm, as well as the Involution Tool~\cite{OMFS20:INTEGRATION},
could be used without any change. Using the CIDM for circuit
modeling is nevertheless superior to using the IDM, because its additional
degree of freedom facilitates a more accurate characterization of the
involved channels w.r.t.\ real circuits. We also note that it is possible
to conduct a proof for the possibility of unbounded SPF directly in the
CIDM (see \cref{sec:possibility}) as well.

\section{Simulating Executions of Circuits}\label{sec:algorithm}

In this section, we provide \cref{algo:iter} for timing simulation in
the CIDM. 
Since our algorithm is supposed to be run in the Involution Tool
\cite{OMFS20:INTEGRATION}, which utilizes
Mentor\textsuperscript{\textregistered}
ModelSim\textsuperscript{\textregistered} (version 10.5c), our implementation
had to be adapted to its internal restrictions.

The general idea is to replace the gates from standard libraries by custom gates
represented by CIDM channels. According to \cref{fig:channel_model_gidm}, a
custom gate $C$ consists of three main components: (i) a pure delay shifter
$P_I=\Deltapm$ for each input $I$ (cancellation is done automatically
by ModelSim), (ii) a
Boolean function representing the embedded gate $G$, and (iii) an IDM channel $c$.
Note that the output of $G$ is in \gls{wst} format, which facilitates direct
comparison of the events ($evGO$) occurring in our CIDM simulation 
with the gate outputs obtained in classic simulations based on standard libraries.

Unfortunately, the input and output of a CIDM channel, i.e., the input of component $P_I$
and the output of component $c$, is of type \gls{tct}, which
is incompatible with discrete event simulations: In ModelSim signal
transitions are represented by events, which are processed in ascending 
order of their scheduled time. Consequently, they cannot be scheduled 
in the past, which may, however, be needed for a transition $(t_n,x_n,o_n)$ with $o_n<0$.

Therefore, for every CIDM channel $C$ we maintain a dedicated file $F(C)$, in
which the simulation algorithm writes all the transitions $(t_n,x_n,o_n)$
generated by $C$. The event $evTI(C)$ that ModelSim inherently assigns to the
output signal of $C$ is only used as a \emph{transition indicator}: For every
edge $(C,I,\Gamma)$, it instantaneously triggers the occurrence of the
``duplicated'' ModelSim event $evTI(C,I,\Gamma)$, which signals to $I$ of
$\Gamma$ the event that there is a new transition in the file $F(C)$. For an
input port $i$, exactly the same applies, except that the input transitions in
$F(i)$ and the transition indicator are externally supplied.

By contrast, both the inputs and output of the gate $G$ embedded in a CIDM
channel $C$, i.e., the output of component $P_I$ and the input of component $c$,
are of type \gls{wst}. Consequently, we can directly use the ModelSim events
$evGI(I,C)$ resp.\ $evGO(C)$ assigned to the output of $P_I$ for input $I$ of
$C$ resp.\ the output of $G$ of $C$ for conveying \gls{wst} transitions. Still,
cancellations that would cause occurrence times $t'<t_{now}$ must be prohibited
explicitly (see \cref{algo:line:addpure} in \cref{algo:iter}).

Our simulation algorithm uses the following functions:

\begin{itemize}
	\item $(t, ev, Par) \gets getNextEvent()$: Returns the event $ev$ with the 
	smallest scheduled time $t$ and
        possibly some additional parameters $Par$. If $t'$ denotes the time of the previous 
event, then $t\geq t'$ is guaranteed to hold. The possible types of events
        are $ev \in \{evTI(\Gamma,I,C), evGI(I,C), evGO(C)\}$, where $C$ is the
channel the event belongs to, $I$ one of its inputs, and $\Gamma$ the vertex
that feeds $I$. If multiple different events are scheduled for the same time $t$, they
occur in the order first $evTI(\Gamma,I,C)$, then $evGI(I,C)$
and finally $evGO(C)$. If multiple instance of the same event are scheduled for the same
time $t$, then only the event that has been scheduled last actually occurs.

	\item $sched(ev, (t, x))$: Schedules a new event $ev$ signaling the transition 
        $x \in \{0,1,toggle\}$ at time $t$; the case $x=toggle$ means $x=1-x'$
        for $x'$ denoting the last (= previous) transition. If the current simulation time is
        $t_{now}$, only $t \geq t_{now}$ is allowed. 

	\item $init(ev, (-\infty, x))$: Initializes the initial state of the event $ev$, 
        without scheduling it.

	\item $value \gets s(ev,t)$: Returns the value of the state function 
        for event $ev\in\{evGI(I,C),evGO(C)\}$ at time $t$. 
            
	\item $F(C) \gets (t, x, o)$: Adds a new \gls{tct} transition 
	$(t,x,o)$ generated by the output of $C$ to the file $F(C)$,
 which buffers the \gls{tct} transitions of $C$.

	\item $(t', x) \gets F(\Gamma)$ reads the most recently added
\gls{tct} transition $(t,x,o)$ from the file $F(\Gamma)$ and returns it as $(t',x)=(t+o,x)$.

	\item $Init(\Gamma)$: For both channels and input ports, the a fixed initial 
value $(-\infty,Init(\Gamma))$.
	
	\item $x \gets G.f(x_1,\dots,x_{G.d})$: Applies the combinatoric function 
        $G.f$ (of the $G.d$-ary gate $G$ embedded in some channel $C$) to the list of logic 
        input values $x_1,\dots,x_{G.d}$, and returns the result $x$.

	\item $postprocess()$: Implements internal bookkeeping. For example, it allows 
to export the \gls{tct} signals contained in our files $F(\Gamma)$ in $\gls{wst}$ format, 
for direct comparison with standard simulation results.
\end{itemize}

\cref{algo:iter} conceptually starts at time $t=-\infty$ and first takes care of
ensuring a clean initial state.  The simulation of the actual execution
commences at $t=0$, where the conceptual ``reset'' that froze the initial state
is released: Every channel whose initial state differs from the computation of
its embedded gate in the initial state causes a corresponding transition of the
gate output at $t=0$, which will be processed subsequently in the main
simulation loop.

\begin{algorithm}
	\footnotesize
	\begin{algorithmic}[1]
          \State \Comment{Executed at simulation time $t=-\infty$:}
		\ForAll{channels $C$} 
			\State $F(C) \gets (-\infty,Init(C),0)$ \Comment{init file}
			\State $init(evTI(C),(-\infty,0))$ \label{line:initTI}\Comment{init toggle indicator}
               		\State $init(evGO(C),(-\infty,Init(C)))$  \label{line:initGO}\Comment{init gate output}
			\ForAll{incoming edges $(\Gamma,I,C)$ of $C$:}
	 			\State $init(evGI(I,C), (-\infty, Init(\Gamma)))$ \label{line:initGI}\Comment{init gate inputs}
			\EndFor
		\EndFor
		\State \Comment{Executed at simulation time $t=0$:}\hfill
		\ForAll{channels $C$, with $d$-ary gate $G$ and incoming edges $(\Gamma_1,I_1,C), \dots, (\Gamma_d,I_d,C)$}
			\State $x = G.f(s(evGI(I_1,C),0),\dots,s(evGI(I_{d},C),0))$
			\If{$x \neq Init(C)$}
                		\State $sched(evGO(C),(0,x))$ \Comment{add reset transition}
			\EndIf
		\EndFor
		\State \Comment{Main simulation loop:}\hfill
		\State $(t, ev, Par) \gets getNextEvent()$			
		\While{$t\leq \tau$}			
		\If{$ev = evTI(\Gamma,I,C)$}\Comment{$evTI$ go first}
		\State $(\Delta^+, \Delta^-, G) \gets Par$
		\State $(t', x) \gets F(\Gamma)$\label{algo:line:readtt}
		\State $\rdmin \gets calcDelta(G.f, x, \Delta^+, \Delta^-)$
		\State $sched(evGI(I,C), (\max\{t,(t' + \rdmin)\},x)$\label{algo:line:addpure}
		
		\ElsIf{$ev = evGI(I,C)$}\Comment{$evGI$ come next}
		\State $(G) \gets Par$
                \State $x \gets s(evGO(C),t)$\Comment{Current gate output}
		\State $y \gets G.f(s(evGI(I_1,C),t),\dots,s(evGI(I_{G.d},C),t))$
		\If{$x\neq y$}\label{line:GOchangecheck}
			\State $sched(evGO(C),(t,y))$
		\EndIf
		\ElsIf{$ev = evGO(C)$}\Comment{and finally $evGO$}
		\State $(x,\dup(.), \ddo(.), \Delta^+, \Delta^-) \gets Par$
		\State $(t', x') \gets F(C)$ 
		\State $T \gets t - t'$
		\If{$x = 1$}
		\State $o \gets \dup(T - \Delta^+) - \Delta^+$
		\Else
		\State $o \gets \ddo(T - \Delta^-) - \Delta^-$
		\EndIf
		\State $F(C) \gets (t,x,o)$
		\State $sched(evTI(C),(t,toggle))$ \Comment{triggers $evTI(C,I,\Gamma')$}
		\EndIf
		\State $(t, ev, Par) \gets getNextEvent()$	
		
		\EndWhile

		\State $postprocess()$\label{algo:line:postproc}
		~\\
		\Procedure{$calcDelta$}{$func, x, \Delta^+, \Delta^-$}
		\If{$func = not$}\
		\If{$x = 1$} \Return $\Delta^-$ \Else\ \Return $\Delta^+$ \EndIf
		\ElsIf{$func = id$} 
		\If{$x = 1$} \Return $\Delta^+$ \Else\ \Return $\Delta^-$ \EndIf
		\Else\
		\State assert($\Delta^+ = \Delta^-$)
		\State \Return $\Delta^+$
		\EndIf
		\EndProcedure
		
	\end{algorithmic}
	\caption{CIDM circuit simulation algorithm}
	\label{algo:iter}
\end{algorithm}


The algorithm uses the procedure $calcDelta$ for computing the actual
$\Delta^+$, $\Delta^-$ employed in the $P_I$ component of a channel, which unfortunately depend
on the embedded gate $G$: For an inverter, for example, a rising transition at
the input leads to a falling transition at the output, so $\ddomin$ must be
used. This is unfortunately not as easy for multi-input gates $G$, since the
effect of any particular input transition \emph{on the output} needs to be known
in advance. This is difficult in the case of an XOR gate, for example, as it depends on the
(future) state of the other inputs. Currently, we hence only support
$\Delta^+ = \Delta^-$ for multi-input gates in our simulation algorithm.

\medskip

Last but not least, we briefly mention a small but important addition to the
Involution Tool, which we have implemented recently. Without knowledge about the 
delay channel, for
example $\dinfty$ or $\fup$/$\fdo$, it is impossible to determine if digital
transition are too close, i.e., that pulses degrade. Note that this is not a
problem of \gls{cidm} alone, but also appears with the widely used inertial delay,
since the user is not able to tell how far away a certain pulse is from removal
(based on the digital results alone, of course).  For this purpose, we added the
possibility to plot the analog trajectory at selected locations of the
circuit. In more detail, we place full-range rising and falling waveforms 
such that the
discretization thresholds are crossed at the predicted points in time. The
switching between between those is, according to our analog channel model, 
instantaneous and thus only continuous in value. Nevertheless, the resulting
waveform provides designers with the possibility to quickly identify 
critical timing issues in the circuit, which may need further investigation 
by means of accurate analog simulation.

\medskip

\cref{thm:simulationcorrectness} below will show that \cref{algo:iter} indeed computes valid executions for a circuit, provided all its logical channels (cp.\
\cref{fig:channel_model_proofs}) are strictly causal. As argued in
\cref{sec:gidmisidm}, this is the case if all interconnected CIDM channels are
compatible (see \cref{def:compatibility}), which will be assumed for the remainder of this section.
Let $\rdmin^C>0$ be the smallest $\rdmin$ of any logical channel in the circuit.

We start with some formal definition for the CIDM.

\bfno{Signals.}  Since we are dealing with two data types for signals in CIDM, 
we need to distinguish \gls{wst} and \gls{tct}. For \gls{wst}, we just re-use the original 
signal definition of the IDM, namely, a list of alternating transitions represented
as tuples $(t_n,x_n)$, $n \geq 0$, where $t_n \in \mathbb{R} \cup \{-\infty\}$
denotes the time when a transition occurs, and $x_n \in \{0,1\}$ whether it is
rising (1) or falling (0). More formally, every \gls{wst} signal $s=((t_n,x_n))_{n\geq 0}$ 
must satisfy the following properties:
\begin{enumerate}
\item[S1)] the initial transition is at time~$t_0=-\infty$.
\item[S2)] the sequence $(x_n)_{n\geq 0}$ must be alternating, i.e., $x_{n+1} = 1- x_n$.
\item[S3)] the sequence $(t_n)_{n\geq 0}$ must be strictly increasing and, if infinite,
	grow without a bound.
\end{enumerate}
To every \gls{wst} signal $s$ corresponds a \emph{state function} 
$s: \mathbb{R}\cup \{-\infty\} \to \{0,1\}$, with
$s(t)$ giving the value of the most recent transition before or at $t$; it
is hence constant in the half-open interval $[t_n,t_{n+1}]$.

For \gls{tct}, we add an offset to the tuples used in \gls{wst}.
A signal here is a list of alternating transitions represented
as tuples $(t_n,x_n,o_n)$, $n \geq 0$, where $t_n \in \mathbb{R} \cup \{-\infty\}$
denotes the time when a transition is \emph{scheduled}, and $x_n \in \{0,1\}$ 
whether it is rising (1) or falling (0). The offset $o_n \in \mathbb{R}\cup \{-\infty\}$
specifies when the transition actually \emph{occurs} in the signal, namely, at 
time $t_n'=t_n + o_n$. Formally, 
every \gls{tct} signal $s=((t_n,x_n,o_n))_{n\geq 0}$ must satisfy the 
following properties: 
\begin{enumerate}
\item[S1)] to S3) are the same as for \gls{wst}.
\item[S4)] $o_0=0$  and the sequence $(t_n')_{n\geq
    2}$ must satisfy $t_{n-2}' \leq t_n'$.
\end{enumerate}

If $t_n'\leq t_{n-1}'$, then we say that the transition $(x_n,t_n,o_n)$
\emph{cancels} the transition $(x_{n-1},t_{n-1},o_{n-1})$. Property S4) implies
that only the most recent, i.e., previous, transition could be canceled. Since
all non-canceled transitions are hence alternating and have strictly increasing
occurrence times, we can define the \gls{tct} state function accordingly: To
every \gls{tct} signal $s$ corresponds a function $s: \mathbb{R} \cup \{-\infty\}\to \{0,1\}$,
with $s(t)$ giving the value of the most recent not canceled transition occurring
before or at $t$. Note that this implies that the \gls{tct} state function
$s_{tct}(t)$ at the input of a cancellation unit $C$ and the \gls{wst} state
function $s_{wst}(t)$ at its output are the same. 

\smallskip

\bfno{Circuits.}  Circuits are obtained by interconnecting a set of input ports
and a set of output ports, forming the external interface of a circuit, and a
set of combinational gates represented by CIDM channels. Recall from \cref{fig:channel_model_gidm} that gates are embedded within a channel here.
We constrain the way
components are interconnected in a natural way by requiring that any channel input
and output port is attached to only one input port or channel output port.

Formally, a {\em circuit\/} is described by a directed graph where:
\begin{enumerate}
\item[C1)] A vertex $\Gamma$ can be either an {\em input port}, an {\em output port}, or
a {\em channel}.
\item[C2)] The \emph{edge} $(\Gamma,I,\Gamma')$ represents a $0$-delay connection from
the output of $\Gamma$ to a fixed input $I$ of $\Gamma'$. 
\item[C3)] Input ports have no incoming edges.
\item[C4)] Output ports have exactly one incoming edge and no outgoing one. 
\item[C5)] A channel that embeds a $d$-ary gate has $d$ inputs $I_1,\dots, I_d$, 
fed by incoming edges, ordered by some fixed order.
\item[C6)] A channel $C$ that embeds a $d$-ary gate $G$ maps $d$ \gls{tct} input
  signals $s_{I_1},\dots, s_{I_d}$ to a \gls{tct} output signal
  $s_C = f_C(s_{I_1},\dots, s_{I_d})$, according to the CIDM channel
  function $f_C(.)$ of $C$ (which also comprises the Boolean gate function
  $G.f$). 
\end{enumerate}

\smallskip

\bfno{Executions.}  An {\em execution\/} of circuit~$\C$ is a collection of
signals~$s_\Gamma$ for all vertices~$\Gamma$ of~$\C$ that respects the channel
functions and input port signals.  Formally, the following
properties must hold:
\begin{enumerate}
\item[E1)] If~$i$ is an input port, then there are no restrictions on~$s_i$.
\item[E2)] If~$o$ is an output port, then~$s_o = s_C$, where~$C$ is the unique
channel connected to~$o$.
\item[E3)] If vertex~$C$ is a channel with~$d$ inputs $I_1,\dots,I_d$, ordered
  according to the fixed order condition C5), and channel function~$f_C$, then
  $s_C = f_C(s_{\Gamma_1},\dots, s_{\Gamma_d})$, where
  $\Gamma_1,\dots,\Gamma_d$ are the vertices the inputs $I_1,\dots,I_d$ of $C$
  are connected to via edges $(\Gamma_1,I_1,C), \dotsm (\Gamma_d,I_d,C)$.
\end{enumerate}

First, we show that all \gls{tct} signals occurring in the CIDM, namely, $u_i^I$, $u_p^I$ 
(and $u_s$, which is, however, the same as $u_i^{I'}$ of a successor channel) in \cref{fig:channel_model_gidm}, also satisfy property S4). Note that all that a pure delay shifter
component $P_I=(\Delta^+,\Delta^-)$ does is to add $\Delta^+$ resp. $\Delta^-$ to the
offset component of every \gls{tct} rising resp.\ falling transition in $u_i^I$ to
generate $u_p^I$.

\begin{lemma}[\gls{tct} correctness]\label{lem:tctcorrectness}
Consider a strictly causal IDM channel fed by a \gls{wst} input 
signal. If its output is interpreted as a \gls{tct} signal, 
it satisfies properties S1)--S4).
\end{lemma}
\begin{proof}
Properties S1)--S3) are immediately inherited from S1)--S3) of the \gls{wst} input 
signal.
 
Assume that $(t_{n+1},x)$, $n\geq 1$, w.l.o.g.\ with $x=0$, is the first transition in the input
signal of $c$ that causes the corresponding output $t_{n+1}'$ to violate S4). 
Let $t_{n-1}'$, $t_n'=t_n+\dup(t_n-t_{n-1}')$ and
  $t_{n+1}'=t_{n+1}+\ddo(t_{n+1}-t_n')$ be the times the corresponding output transitions
  occur. By our assumption, $t_{n+1}' < t_{n-1}'$. Since
  $t_{n+1}\geq t_n$, we find by using monotonicity of $\ddo(.)$ and the involution
  property that
\begin{align}
t_{n+1}' &= t_{n+1} + \ddo(t_{n+1}-t_n-\dup(t_n-t_{n-1}')) \nonumber\\
&\geq t_{n+1} + \ddo(-\dup(t_n-t_{n-1}'))\nonumber\\
     &= t_{n+1}-t_n+t_{n-1}' \geq t_{n-1}'\nonumber,
\end{align}
which provides the required contradiction. The proof for $x=1$ is analogous.
\end{proof}

\cref{lem:tctcorrectness} in conjunction with the strict causality established
for consistent CIDM channels in \cref{sec:gidmisidm} reveals that we can indeed
use the inherent cancellation of out-of-order ModelSim events for implementing 
the cancellation unit $C$ in \cref{algo:line:addpure} of \cref{algo:iter}.
Note that we must explicitly prohibit scheduling a canceling event before 
$t_{now}$, however.


The following \cref{lem:evGOseq} shows that the gate input signals 
$evGI(I,C)$ and gate output signals $evGO(C)$ as well as
every $evTI(\Gamma,I,C)$ are indeed of type \gls{wst}. 
This result is mainly implied by be internal workings of ModelSim
and strict causality.

\begin{lemma}[\gls{wst} correctness]\label{lem:evGOseq}
Every signal $evTI(\Gamma,I,C)$, $evGI(I,C)$ and $evGO(C)$ maintained by \cref{algo:iter}
for a circuit comprising only consistent CIDM channels is of type \gls{wst}.
\end{lemma}
\begin{proof}
Condition S1) follows from \cref{line:initTI}--\cref{line:initGI}. 
Condition S2) follows for $evTI$ from using the toggle mode and
for $evGO$ from \cref{line:GOchangecheck}. 
For $evGI$, it is implied by the fact that the signal 
$u_i^I$ of $C$ maintained via $F(\Gamma)$ and $evTI(\Gamma,I,C)$ 
is alternating in the absence of cancellations, as it is generated by the 
alternating signal $evGO$ in $\Gamma$. Since, by \cref{lem:tctcorrectness}, cancellations
do not destroy this alternation for consistent CIDM channels either, the claim follows.

As far as condition S3) is concerned, $t_{n+1}-t_n > 0$ for $n\geq 0$ follows from 
the inherent properties of ModelSim, which ensure that only the
last of several instances of the same event $ev$ scheduled for the same time $t$
actually occurs. However, we also need to prove that $t_n$ grows without bound
if there are infinitely many transitions in the signal. We first consider
$evGO(C)$ and assume, for a contradiction, that the limit of the monotonically
increasing sequence of occurrence times is $\lim_{n\to\infty} t_n = L < \infty$.

Let $C=C_0$. By the pigeonhole principle, there must be an input $I$ of $C_0$ 
which causes a subsequence of infinitely many of the transitions $evGO(C_0)$. Consider the channel 
$C_1$ which feeds this $I$, and consider any $t_n^0$. This transition is caused by the, say, $k$-th
transition of $evGI(I,C_0)$ and hence of $evTI(C_1,I,C_0)$, which also occur at time 
$t_n^0$. Obviously, their cause is the $k$-th transition of $evGO(C_1)$, which occurs at time 
$t_k^1$. Clearly, the latter did not cancel the $k-1$-st transition of $evTI(C_1,I,C_0)$,
otherwise the $k$-th transition of $evTI(C_1,I,C_0)$ would not have happened either. From
\cite[Lem.~4]{FNNS19:TCAD}, it follows that $t_n^0 - t_k^1 > \rdmin \geq \rdmin^C$,
where $\rdmin$ is the one of the logical channel between the gate of $C_1$ and the gate of 
$C_0$ (which is an IDM channel).
Consequently, all the infinitely many transitions of $evGO(C_1)$ that cause the 
subsequence of transitions of $evGO(C_0)$ must actually occur by time $L-\rdmin^C$.

If we delete all the transitions in the signal $evGO(C_1)$ after time $L-\rdmin^C$,
we can repeat exactly the same argument inductively. Repeating this $i$ times, we
end up with some channel
$C_i$ with infinitely many transitions of $evGO(C_i)$ by time $L-i\rdmin^C$.
Clearly, for $i > L/\rdmin^C$, this is impossible, which provides the required
contradiction.

Since infinitely many transitions of either $evTI(\Gamma,I,C)$ and $evGI(I,C)$
in finite time can be traced back to infinitely many transitions in finite time of
$evGO(\Gamma)$, the above contradiction also rules out these possibilities, which
completes our proof.
\end{proof}

We get the following corollaries:

\begin{lemma}[File consistency]\label{lem:consistencyTIandF}
For every channel $C$, $evTI(C)=evTI(C,I,\Gamma)$ and $F(C)$ can never become 
inconsistent, i.e., all successors $\Gamma$ of $C$ can read $F(C)$ when processing $evTI(C,I,\Gamma)$
in \cref{algo:line:readtt} before $evTI(C)$ is toggled again.
\end{lemma}

\begin{proof}
Assume that the processing of $evGO(C)$ of channel $C$, at time $t$, 
writes a new transition to $F(C)$ and toggles $evTI(C)$. 
This causes the scheduling of $evTI(C,I,\Gamma)$ for time $t$ 
for $\Gamma$. In order to overwrite $F(C)$ before $\Gamma$ could read it, 
a canceling transition must have been written to $F(C)$ by time $t$.
This cannot happen strictly before $t$, though, as then both the original $evGO(C)$
and $evTI(C)$ would have been canceled. It cannot happen at $t$
either, since this is prohibited by \cref{lem:consistencyTIandF}. Hence,
we have established the required contradiction.
\end{proof}

\begin{theorem}[Correctness of simulation]\label{thm:simulationcorrectness}
For any $0 \leq \tau < \infty$, the simulation \cref{algo:iter}
applied to a circuit with compatible CIDM channels always
terminates with a unique execution up to time $\tau$.
\end{theorem}

\begin{proof}
According to \cref{lem:evGOseq}, the simulation time $t$ of the main
simulation loop grows without bound, unless the execution consists of
finitely many events only. By \cref{lem:tctcorrectness} and 
\cref{lem:consistencyTIandF}, the simulation algorithm correctly simulates
the behavior of CIDM channels and thus indeed generates an execution $E$ of our circuit.
To show that $E$ is unique, assume that there is another execution
$E'$ that satisfies E1)--E3). As it must differ from $E$ in some transition
generated for some channel $C$, $E'$ cannot adhere to the correct 
channel function $f_C$, hence violates E3).
\end{proof}

\section{Experiments}
\label{sec:experiments}

In this section, we validate our theoretical results by means of simulation
experiments. This requires two different setups: (i) To validate the CIDM, we
incorporated \cref{algo:iter} in our Involution Tool 
\cite{OMFS20:INTEGRATION} and compared its predictions to other
models.
(ii) To establish the mandatory prerequisite for these experiments, namely, an
accurate characterization of the delay functions, we employed a fairly elaborate
analog simulation environment.

\begin{sloppypar}
Relying on the \SI{15}{\nm} Nangate Open Cell Library featuring
FreePDK15$^\text{TM}$ FinFET models~\cite{Nangate15} ($\vdd=\SI{0.8}{\V}$), we
developed a Verilog description of our circuits and used the
Cadence\textsuperscript{\textregistered} tools Genus\textsuperscript{TM} and
Innovus\textsuperscript{TM} (version 19.11) for optimization, placement and
routing. We then extracted
the parasitic networks between gates
from the final layout, which resulted in accurate \spice\ models that were simulated
with Spectre\textsuperscript{\textregistered} (version 19.1). These
results were used both for gate characterization and as a golden
reference for our digital simulations.
\end{sloppypar}

Like in \cite{FNNS19:TCAD}, our main target circuit is a custom inverter
chain. In order to highlight the improved modeling accuracy of CIDM, it consists
of seven alternating high- and low-threshold inverters. They were implemented by
increasing the channel length of p- respectively nMOS transistors, which varies
the transistor threshold voltages \cite[Fig.~2]{Ase98}.
For comparison purposes, we conducted experiments with a standard
inverter chain as well.

Regarding gate characterization for \gls{idm}, we used two different
approaches. Recall from \cref{obs:singlefixesall} that fixing a single
discretization threshold pins the value of all consistent $\dmin$, $\vthin$ and
$\vthout$ throughout the circuit. In the variant of \gls{idm} called \idmm, we
chose $\vthoutm=\vdd/2$ for the last inverter in the chain, and determined the
actual value of its matching $\vthinm$ by means of analog simulations. To obtain
consistent discretization thresholds for the whole circuit, we repeated this
characterization, starting from $\vthoutm=\vthinm$ for the next inverter up the
chain. We thereby obtained values in the range $[0.301,0.461]\ \si{\V}$, with
$\vthinm=\SI{0.455}{\V}$ for the first gate.  Obviously, characterizing a
circuit in this fashion is very time-consuming, as only a single gate in a path
can be processed at a time.

\begin{figure}[tb]
	\centering
	\includegraphics[width=0.98\columnwidth]{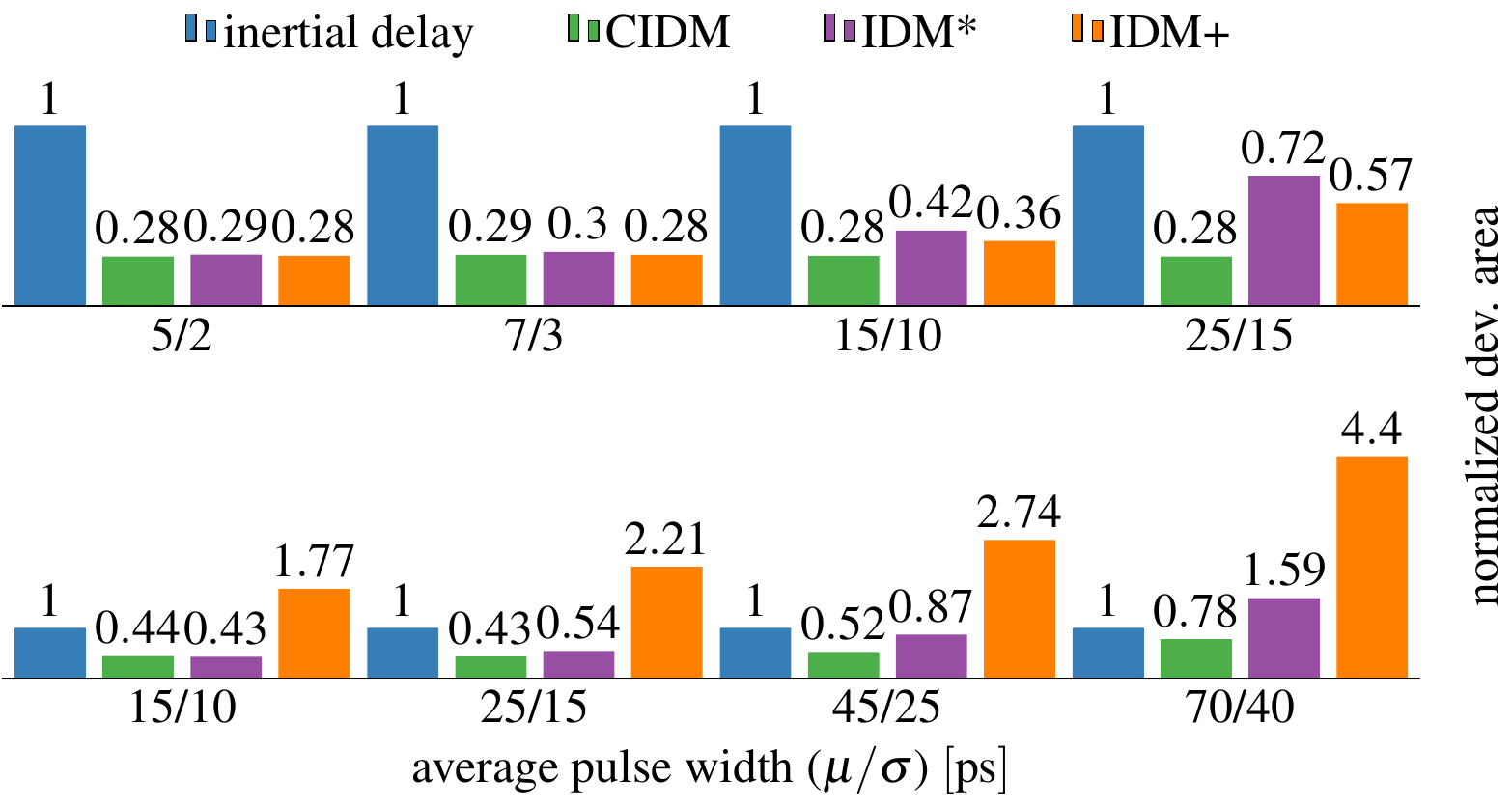}
	\caption{Accuracy, expressed as the normalized total deviation area of
		the digital predictions, relative to \spice\ for the standard inverter 
		chain (top) and high/low threshold inverter chain (bottom). Lower bars indicate
		better results.}
	\label{fig:results_area_comb}
\end{figure}

  Moreover, forks (that is, joins) are problematic for this characterization,
  since the input characterization thresholds of two of its end most certainly
  do not coincide. Reversing the direction of characterization, i.e., starting
  at the front and propagating towards the end, would solve this problem but
  adds a similar difficulty at the inputs of multi-input gates. Needless to say,
  feedback loops most probably make any such attempt impossible.

Alternatively, we also separately characterized every gate for $\vthoutm=\vdd/2$ and determined
the matching $\vthinm$, which we will refer to as \idmf.  Note carefully that
the discretization thresholds of connected gate out- and inputs differ for
\idmf, such that an error is introduced at every interconnecting edge.

Since the signals are very steep
near $\vdd/2$, the introduced error is typically rather small. This is even
more pronounced due to the fact that the deviation of the input thresholds 
is usually smaller than that of the output thresholds due
to the natural amplification of a gate. Note that this was verified by our 
simulations of the standard inverter chain.
However, although the misprediction is small, it is introduced for each
transition at every gate.  While this might be negligible for small circuits
like our chain, the error quickly accumulates for larger devices leading to
deviations even for very broad pulses. Thus, the \idmm\ can be expected to
deliver worse results than pure/inertial delay while being a computationally
much more expensive approach. Indeed, for the gates used in our standard
inverter chain, we recognized a clear bias towards $\vthinm<\vdd/2$ for
$\vthoutm=\vdd/2$. Results for these delay functions are shown in 
\cref{fig:results_area_comb}.
Finally, characterizing gates for \gls{cidm} was simply executed for $\vthout=\vthin=\vdd/2$.


The results for stimulating the standard inverter chain, with 2,500 normally distributed
pulses of average duration $\mu$ and standard deviation $\sigma$, obtained by the Involution Tool for 
\idmm, \idmf, \gls{cidm} and the default inertial delay model, are shown in
\cref{fig:results_area_comb} (top). The accuracy of the model predictions are
presented relative to our digitized \spice\ simulations, which gets
subtracted from the trace obtained with a digital delay model.  Summing up the
area (without considering the sign), we obtain a metric that can be used to
compare the similarity of two traces.  Since the absolute values of the area are
inexpressive, we normalize the results and use the inertial delay model as
baseline.

For short pulses, \idmm, \idmf\ and \gls{cidm} perform similarly.  We conjecture
that this is a consequence of the narrow range for $\vthoutm$ and $\vthinm$
($[0.39156, 0.4]\ \si{\V}$), and therefore the induced error due to non-perfect
matchings in \idmf\ is negligible.  For broader pulses, we observe a reduced
accuracy of \idmm\ and \idmf, which is primarily an artifact of the imperfect
approximation of the real delay function by the ones supported by the Involution
Tool.  We even observed settings, where \gls{cidm} does not even beat the
inertial delay model, which can also be traced to this cause.

For our custom inverter chain [\cref{fig:results_area_comb} (bottom)], \gls{cidm}
outperforms, as expected, the other
models considerably, whereas \idmf\ occasionally delivers poor results, even compared to
inertial delays. This is a consequence of the non-matching threshold
values and the accumulating error. \idmm\ achieves much better predictions, but
still falls short compared to \gls{cidm}. For broader pulses, \gls{cidm} and the
inertial delay model perform similar, since they use the same maximum delay
$\dup(\infty)$ and $\ddo(\infty)$. The degradation of \idmm\ is once again a
result of the imperfect delay function approximations.

Finally, analog simulations in \cref{fig:experimental_results} revealed that an
oscillation slightly below $\vdd/2$ at the input of a low-threshold inverter can
still result in full range switches at the end of the chain. For the fast
\gls{idm} characterization such traces get removed. Note that albeit the
presented oscillation is visible in \gls{idm} when characterized from back to
front there are still infinitely many possibilities that can not be
detected. Even if these do not propagate further it is important to know if the
circuit has stabilized or not. The digital simulation results are shown in
\cref{fig:experimental_results}. While \gls{idm} removes the oscillation and
thus does not indicate any transition at the output \gls{cidm} correctly
predicts the regeneration of the pulses.

Note that we also added the capability to display the analog
waveform $u_r$ that corresponds to the digital threshold crossings to
the Involution Tool. This makes
it possible for the user to investigate if the pulses are ill shaped and thus
the circuit needs improvements.

To summarize the results of our experiments, we highlight that the
characterization procedure for \gls{idm} either requires high effort (\idmm)
or may lead to modeling inaccuracies (\idmf). 
The \gls{cidm} clearly outperforms
all other models w.r.t.\ modeling accuracy for our custom inverter chain, and
is also the only model that can faithfully predict the ``de-cancellation'' of
sub-threshold pulses.

\begin{figure}[tb]
  \centering
  \includegraphics[width=0.98\columnwidth]{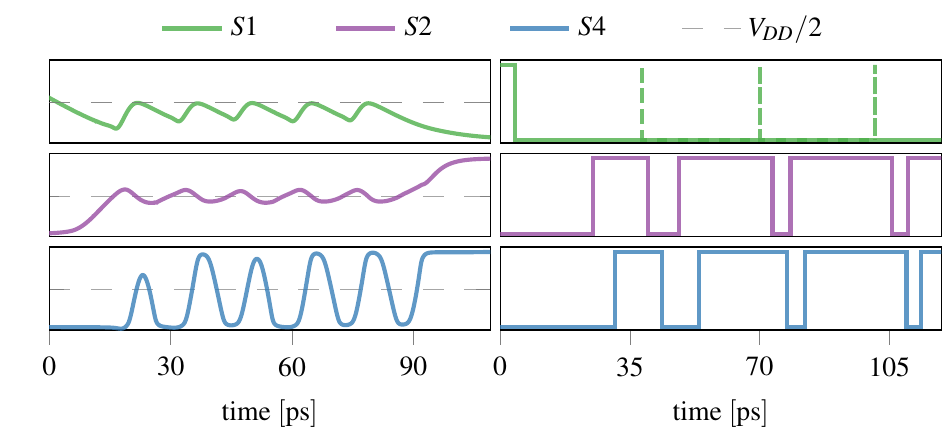}
  \caption{Recovering sub-threshold waveforms in an inverter chain using the \gls{cidm}.}
  \label{fig:experimental_results}
\end{figure}

\section{Conclusions}
\label{sec:conclusion}

We presented the \glsentryfull{cidm}, a generalization of the 
\glsentryfull{idm} that retains its faithful glitch-propagation
properties. Its distinguishing properties are wider applicability,
composability, easier characterization of the delay functions, and exposure of
canceled pulse trains at interconnecting wires. 
The CIDM and
our novel digital timing simulation algorithm have been 
developed on sound theoretical foundations, which allowed us to 
rigorously prove their properties. Analog and digital simulations 
for inverter chains were used to confirm our theoretical predictions.
Despite this considerable step forward towards a faithful delay model,
there is still some room for improvement, in particular, for accurately
modeling the delay of multi-input gates. 


\end{document}